\let\oldvec\vec
\documentclass{llncs}
\let\vec\oldvec
\usepackage{amssymb}
\usepackage{graphicx}
\usepackage{fullpage}
\usepackage{bm}
\usepackage{mathtools}
\usepackage{color}
\usepackage{array,graphicx}
\usepackage{longtable}
\usepackage{changepage}
\usepackage{blkarray}
\usepackage{scalerel}
\usepackage{hyperref}
\usepackage{algorithm2e}
\usepackage[title]{appendix}
\usepackage{pgf,pgfarrows,pgfnodes}
\usepackage{tikz}
\usepackage[caption=false,font=footnotesize,farskip=0pt]{subfig}
\usepackage[left=1.5in,top=1.5in,right=1.5in,bottom=1.5in,nohead]{geometry}
\usepackage{url,latexsym,amsmath,amssymb, amsfonts,enumerate, graphics,times,commath}
\usepackage[utf8]{inputenc}
\usepackage[version=3]{mhchem}
\usepackage{multirow}
\usepackage{cite}

\usetikzlibrary{snakes}
\usetikzlibrary{automata,positioning,arrows}

\renewcommand{\S}{S}

\newcommand{\R}{R}

\makeatletter
\newcommand{\@chapapp}{\relax}%

\makeatother
{\begin{list}
    {\noindent\makebox[0mm][r]{$\bullet$}}
    {\leftmargin=5.5ex \usecounter{enumi}
      \topsep=1.5mm \itemsep=-.75ex}
  }
  {\end{list}}

\title{A reaction network scheme which implements inference and learning for Hidden Markov Models}
\author{Abhinav Singh\inst{1}\and Carsten Wiuf\inst{2} \and Abhishek Behera\inst{3} \and Manoj Gopalkrishnan\inst{3}}
\institute{UM-DAE Centre for Excellence in Basic Sciences, Mumbai, India\\ \and
  Department of Mathematical Sciences, University of Copenhagen, Denmark\\
  \and Indian Institute of Technology Bombay, Mumbai, India\\	
  \email{\{abhinavsns7, abhishek.enlightened, manoj.gopalkrishnan\}@gmail.com \\  wiuf@math.ku.dk}}
\date{7 March 2019}
\begin{document}
\maketitle
\begin{abstract}
  With a view towards molecular communication systems and molecular multi-agent systems, we propose the Chemical Baum-Welch Algorithm, a novel reaction network scheme that learns parameters for Hidden Markov Models (HMMs). Each reaction in our scheme changes only one molecule of one species to one molecule of another. The reverse change is also accessible but via a different set of enzymes, in a design reminiscent of futile cycles in biochemical pathways. We show that every fixed point of the Baum-Welch algorithm for HMMs is a fixed point of our reaction network scheme, and every positive fixed point of our scheme is a fixed point of the Baum-Welch algorithm. We prove that the ``Expectation'' step and the ``Maximization'' step of our reaction network separately converge exponentially fast. We simulate mass-action kinetics for our network on an example sequence, and show that it learns the same parameters for the HMM as the Baum-Welch algorithm.
\end{abstract}
\section{Introduction}
The sophisticated behavior of living cells on short timescales is powered by biochemical reaction networks. One may say that evolution has composed the symphony of the biosphere, genetic machinery conducts the music, and reaction networks are the orchestra. Understanding the capabilities and limits of this molecular orchestra is key to understanding living systems, as well as to engineering molecular systems that are capable of sophisticated life-like behavior. 

The technology of implementing abstract reaction networks with molecules is a subfield of molecular systems engineering that has witnessed rapid advances in recent times. Several researchers  
\cite{soloveichik2010dna,srinivas2015programming,qian2011efficient,Cardelli_2011StrandAlgebra,lakin2011abstractions,cardelli2013two,chen2013programmable,LAKIN201621} 
have proposed theoretical schemes for implementing arbitrary reaction networks with DNA oligonucleotides. There is a growing body of experimental demonstrations of such schemes 
\cite{Srinivasetal2052,cherry2018scaling,srinivas2015programming,chen2013programmable,zechner2016molecular}. 
A stack of tools is emerging to help automate the design process. We can now compile abstract reaction networks to a set of DNA oligonucleotides that will implement the dynamics of the network in solution~\cite{badelt2017general}. We can computationally simulate the dynamics of these oligonucleotide molecular systems~\cite{lakin2011visual} to allow debugging prior to experimental implementation. In view of these rapid advances, the study of reaction networks from the point of view of their computational capabilities has become even more urgent. 

It has long been known that reaction networks can compute any computable function~\cite{hjelmfelt1991chemical}. The literature has several examples of reaction network schemes that have been inspired by known algorithms~\cite{buisman2009computing,Klavins_2011Biomolecular,soloveichik2008computation,chen2014deterministic,qian2011scaling,napp2013message,Winfree_2011Neural,cardelli2018chemical}. Our group has previously described reaction network schemes that solve statistical problems like maximum likelihood~\cite{gopalkrishnan2016scheme}, sampling a conditional distribution and inference~\cite{virinchi2017stochastic}, and learning from partial observations~\cite{virinchi2018reaction}. These schemes exploit the thermodynamic nature of the underlying molecular systems that will implement these reaction networks, and can be expressed in terms of variational ideas involving minimization of Helmholtz free energy~\cite{Amari2016,csiszar2003information,jaynes1957information}.

In this paper, we consider situations where partial information about the environment is available to a cell in the form of a sequence of observations. For example, this might happen when an enzyme is acting processively on a polymer, or a molecular walker~\cite{shin2004synthetic,Reif_2003WalkerRoller,Seeman_2004Walker} is trying to locate its position on a grid. In situations like this, multiple observations are not independent. Such sequences can not be summarized merely by the {\em type} of the sequence~\cite{cover2012elements}, i.e., the number of times various symbols occur. Instead, the order of various observations carries information about state changes in the process producing the sequence. The number of sequences grows exponentially with length, and our previously proposed schemes are algorithmically inadequate. To deal with such situations requires a pithy representation of sequences, and a way of doing inference and learning directly on such representations. In Statistics and Machine Learning, this problem is solved by \textbf{Hidden Markov Models (HMMs)}~\cite{18626}.

HMMs are a widely used model in Machine Learning, powering sequence analysis applications like speech recognition~\cite{juang1991hidden}, handwriting recognition, and bioinformatics. They are also essential components of communication systems as well as of intelligent agents trained by reinforcement learning methods. In this article, \textbf{we describe a reaction network scheme which implements the Baum-Welch algorithm}. The Baum-Welch algorithm is an iterative algorithm for learning HMM parameters. Reaction networks that can do such statistical analysis on sequences are likely to be an essential component of molecular communication systems, enabling cooperative behavior among a population of artificial cells. Our main contributions are:
\begin{enumerate}
\item In Section~\ref{sec:hmm}, we describe what the reader needs to know about HMMs and the Baum-Welch algorithm to be able to follow the subsequent constructions. No prerequisites are assumed beyond familiarity with matrices and probability distributions.
\item In Section~\ref{sec::reac}, we describe a novel reaction network scheme to learn parameters for an HMM. 
\item We prove in Theorem~\ref{thm:bw} that every fixed point of the Baum-Welch algorithm is also a fixed point of the continuous dynamics of this reaction network scheme. 
\item In Theorem~\ref{thm:bwcrn}, we prove that every positive fixed point of the dynamics of our reaction network scheme is a fixed point of the Baum-Welch algorithm. 
\item In Theorem~\ref{th:EMexpconv1} and Theorem~\ref{th:EMexpconv2}, we prove that subsets of our reaction network scheme which correspond to the Expectation step and the Maximization step of the Baum-Welch algorithm both separately converge exponentially fast.
\item In Example~\ref{Ex::1}, we simulate our reaction network scheme on an input sequence and show that the network dynamics is successfully able to learn the same parameters as the Baum-Welch algorithm.
\item We show in Example~\ref{Ex::2} that when the baum-welch fixed point is on the boundary then our scheme need not converge to a Baum-Welch fixed point. However, we conjecture that if there exists a positive Baum-Welch fixed point then our scheme must always converge to a Baum-Welch fixed point. In particular, there would be a positive Baum-Welch fixed point if the true HMM generating the sequence has all parameters positive, and the observed sequence is long enough.
\end{enumerate}

\section{Hidden Markov Models and the Baum Welch Algorithm}\label{sec:hmm}

\begin{figure}[h!]\label{fig:hmm}
  \centering
  \subfloat[\textbf{Hidden Markov Model}]{
    \raisebox{15pt}
    {
      \def\svgscale{1.0}
      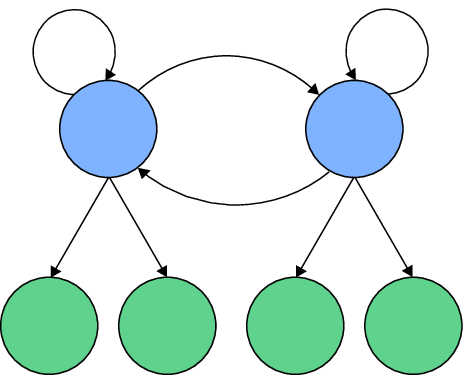}
    \label{fig:hmm1}
  }
  \subfloat[\textbf{Forward Algorithm}]{{
      \def\svgscale{1}
      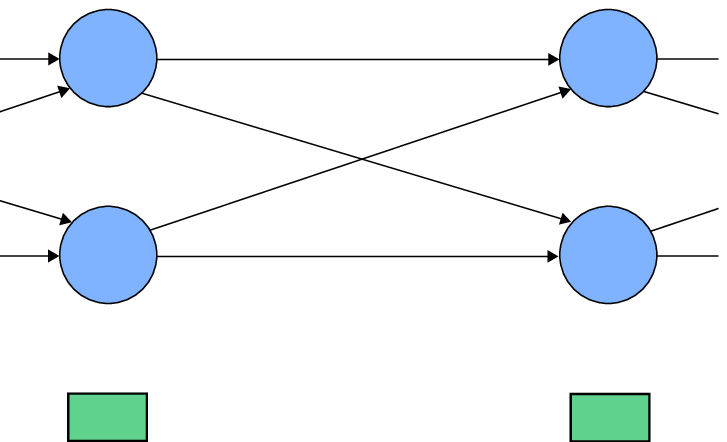\label{fig:forward} }}\
  \vspace{3mm}
  \\\subfloat[\textbf{Backward Algorithm}]{{
      \def\svgscale{1}
      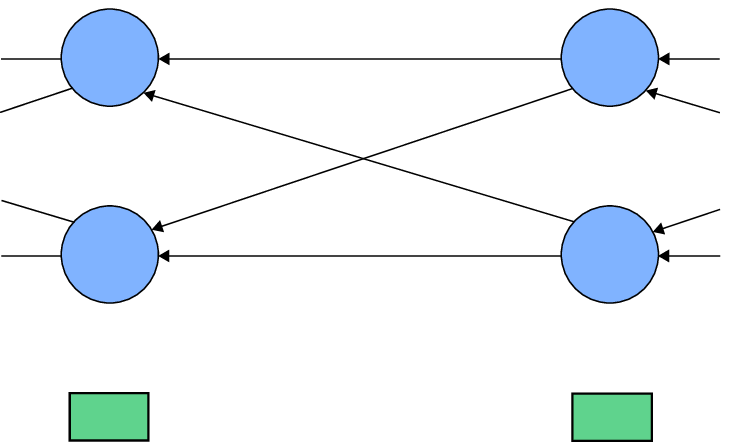\label{fig:backward} }}
  \subfloat[\textbf{Baum-Welch Algorithm}]{{
      \def\svgscale{.95}
      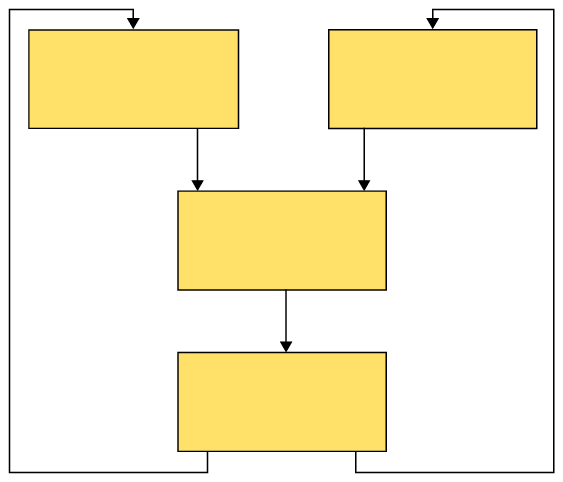\label{fig:bw} }}		
  \caption{\textbf{Learning HMMs from sequences.} (a) \textbf{HMM}: The hidden states $H_1$ and $H_2$ are not directly observable. Instead what are observed are elements $V_1,V_2$ from the set $V=\{V_1,V_2\}$ of ``visible states.'' The parameters $\theta_{11},\theta_{12},\theta_{21},\theta_{22}$ denote the probability of transitions between the hidden states. The probability of observing states $V_1,V_2$ depends on the parameters $\psi_{11},\psi_{12},\psi_{21},\psi_{22}$ as indicated in the figure. 
    (b) The \textbf{forward algorithm} computes the position $l+1$ likelihood $\alpha_{l+1,1} = \alpha_{l1}\theta_{11}\psi_{1v_{l+1}} +  \alpha_{l2}\theta_{21}\psi_{1v_{l+1}}$ by forward propagating the position $l$ likelihoods $\alpha_{l1}$ and $\alpha_{l2}$. Here $v_l,v_{l+1}\in V$ are the observed emissions at position $l$ and $l+1$.
    (c) The \textbf{backward algorithm} computes the position $l-1$ conditional probability $\beta_{l-1,1} = \theta_{11}\psi_{1v_l}\beta_{l1}+\theta_{12}\psi_{2v_l}\beta_{l2}$ by propagating the position $l$ conditional probabilities $\beta_{l1}$ and $\beta_{l2}$ backwards.
    (d) The \textbf{Baum-Welch Algorithm} is a fixed point Expectation-Maximization computation. The E step calls the forward and backward algorithm as subroutines and, conditioned on the entire observed sequence $(v_1,v_2,\dots,v_L)\in V^L$, computes  the probabilities $\gamma_{lg}$ of being in states $g\in H$ at position $l$  and the probabilities $\xi_{lgh}$ of taking the transitions $gh\in H^2$ at position $l$. The M step updates the parameters $\theta$ and $\psi$ to maximize their likelihood given the observed sequence.}
\end{figure}

Fix two finite sets $P$ and $Q$. A \textbf{stochastic map} is a $|P|\times |Q|$ matrix $A=(a_{pq})_{|P|\times |Q|}$ such that $a_{pq}\geq 0$ for all $p\in P$ and $q\in Q$, and $\sum_{q\in Q} a_{pq} = 1$ for all $p\in P$. Intuitively, stochastic maps represent conditional probability distributions. 

An \textbf{HMM} $(H,V,\theta,\psi,\pi )$ consists of finite sets $H$ (for `hidden') and $V$ 
(for `visible'), a  stochastic map $\theta$ from $H$ to $H$ called the \textbf{transition 
  matrix}, a stochastic map $\psi$ from $H$ to $V$ called the \textbf{emission matrix}, and 
an \textbf{initial probability distribution} $\pi = (\pi_h)_{h\in H}$ on $H$, i.e., $\pi_h
\geq 0$ for all $h\in H$ and $\sum_{h\in H} \pi_h = 1$. See Figure~\ref{fig:hmm1} for an 
example.

Suppose a length $L$ sequence $(v_1,v_2,\dots,v_L)\in V^L$ of visible states is observed due to a hidden sequence $(x_1,x_2, \dots,x_L)\in H^L$. The following questions related to an HMM are commonly studied:
\begin{enumerate}
\item \textbf{Likelihood:} For fixed $\theta, \psi$, compute the likelihood $\Pr(v_1,v_2,\dots,v_L\mid \theta,\psi)$. This problem is solved by the \textbf{forward-backward algorithm}.
\item \textbf{Learning:} Estimate the parameters $\hat\theta,\hat\psi$ which maximize the likelihood of the observed sequence $(v_1,v_2,\dots,v_L)\in V^L$. This problem is solved by the \textbf{Baum-Welch algorithm} which is an Expectation-Maximization (EM) algorithm. It uses the forward-backward algorithm as a subroutine to compute the E step of EM.
\item  \textbf{Decoding:} For fixed $\theta,\psi$, find the sequence $(\hat h_1,\hat h_2,\dots,\hat h_l,\dots,\hat h_L)\in H^L$ that has the highest probability of producing the given observed sequence $(v_1,v_2,\dots,v_L)$. This problem is solved by the \textbf{Viterbi algorithm}.
\end{enumerate}

The \textbf{forward algorithm}~(Figure~\ref{fig:forward}) takes as input an HMM $(H,V, \theta,\psi,\pi)$ and a length $L$ observation sequence $(v_1,v_2,\dots,v_L) \in V^L$ and outputs the $L\times |H|$ likelihoods $\alpha_{l h} = \Pr[v_1,v_2,\dots,v_l,x_l=h \mid \theta,\psi]$ of observing symbols $v_1,\ldots,v_l$ and being in the hidden state $h\in H$ at time $l$. It does so using the following recursion.
\begin{itemize}
\item Initialisation: $\alpha_{1h} = \pi_h\psi_{hv_1}$ for all $h\in H$,
\item Recursion: $\alpha_{l h}=\sum_{g\in H} \alpha_{l-1,g}\theta_{gh}\psi_{h v_l}$ for all $h\in H$ and $l=2,\ldots,L$.
\end{itemize} 

The \textbf{backward algorithm}~(Figure~\ref{fig:backward}) takes as input an HMM $(H,V, \theta,\psi,\pi)$ and a length $L$ observation sequence $(v_1,v_2,\dots,v_L) \in V^L$ and outputs the $L\times |H|$ conditional probabilities\
\\ $\beta_{l h} = \Pr[v_{l+1},v_{l+2},\dots,v_L \mid x_l=h, \theta,\psi]$ of observing symbols $v_{l+1},\ldots,v_L$ given that the hidden state $x_l$ at time $l$ has label $h\in H$.
\begin{itemize}
\item Initialisation: $\beta_{L h}=1$, \quad for all $h\in H$,
\item Recursion: $\beta_{l h}=\sum_{g\in H} \theta_{hg}\psi_{g v_{l+1}}\beta_{l+1,g}$ for all $h\in H$ and $l=1,\ldots,L-1$.
\end{itemize} 

The \textbf{E step} for the Baum-Welch algorithm takes as input an HMM $(H,V, \theta,\psi,\pi)$ and a length $L$ observation sequence $(v_1,v_2,\dots,v_L)$. It outputs the $L\times H$ conditional probabilities $\gamma_{lh}=\Pr[x_l = h \mid \theta,\psi, v]$ of being in hidden state $h$ at time $l$ conditioned on the observed sequence $v$ by:
\[
  \gamma_{lh}=\frac{\alpha_{l h}\beta_{l h}}{\sum_{g\in H} \alpha_{l g}\beta_{l g}}
\]
for all $h\in H$ and $l=1,2,\dots,L-1$. It also outputs the $(L-1)\times H \times H$ probabilities $\xi_{lgh}= \Pr[x_l=g, x_{l+1}=h \mid \theta,\psi, v]$ of transitioning along the edge $(g,h)$ at time $l$ conditioned on the observed sequence $v$ by:
\[
  \xi_{lgh}=\frac{\alpha_{l g}\theta_{gh}\psi_{ hv_{l+1}}\beta_{l+1,h}}{\sum_{f\in H} \alpha_{l f}\beta_{l f}}
\]
for all $ g,h\in H$ and $l=1,\ldots,L-1$. 

\begin{remark}\label{rmk:scale}
  Note that the E-step uses the forward and backward algorithms as subroutines to first compute the $\alpha$ and $\beta$ values. Further note that we don't need the forward and backward algorithms to return the actual values $\alpha_{lh}$ and $\beta_{lh}$. To be precise, let $\alpha_l=(\alpha_{lh})_{h\in H}\in\mathbb{R}^H$ denote the vector of forward likelihoods at time $l$. Then for the E step to work, we only need the direction of $\alpha_l$ and not the magnitude. This is because the numerator and denominator in the updates are both linear in $\alpha_{lh}$, and the magnitude cancels out. Similarly, if $\beta_l=(\beta_{lh})_{h\in H}\in\mathbb{R}^H$ denotes the vector of backward likelihoods at time $l$ then the $E$ step only cares about the direction of $\beta_l$ and not the magnitude. This scale symmetry is a useful property for numerical solvers. We will also make use of this freedom when we implement a lax forward-backward algorithm using reaction networks in the next section.
\end{remark}

The \textbf{M step} of the Baum-Welch algorithm takes as input the values $\gamma$ and $\xi$ that are output by the E step as reconstruction of the dynamics on hidden states, and outputs new Maximum Likelihood estimates of the parameters $\theta,\psi$ that best explain these values. The update rule turns out to be very simple. For all $g,h\in H$ and $w\in V$:
\begin{align*}
  \theta_{gh}\leftarrow\frac{\sum_{l=1}^{L-1}\xi_{lgh}}{\sum_{l=1}^{L-1}\sum_{f\in H}\xi_{lgf}},
  &&\psi_{hw}\leftarrow\frac{\sum_{l=1}^L\gamma_{lh}\delta_{w,v_l}}{\sum_{l=1}^L\gamma_{lh}}
\end{align*}
where $\delta_{w,v_l}=\begin{cases}
  1\text{ if }w=v_l\
  \\0\text{ otherwise}\end{cases}$ is the Dirac delta function.

\begin{remark}\label{rmk:scale2}
  Like in Remark~\ref{rmk:scale}, note that the M step does not require its inputs to be the actual values $\gamma$ and $\xi$. There is a scaling symmetry so that we only need the directions of the vectors $\gamma(h) = (\gamma_{lh})_{l=1,2,\dots,L}\in\mathbb{R}^L$ for all $h\in H$ and $\xi(g) = (\xi_{lgh})_{1\leq l\leq L-1, h\in H}\in\mathbb{R}^{(L-1)\times H}$ for all $g\in H$. This gives us the freedom to implement a lax E projection without affecting the M projection, and we will exploit this freedom when designing our reaction network.
\end{remark}

The \textbf{Baum-Welch algorithm}~(Figure~\ref{fig:bw}) is a fixed point EM computation that alternately runs the E step and the M step till the updates become small enough. It is guaranteed to converge to a fixed point $(\hat\theta,\hat\psi)$. However, the fixed point need not always be a global optimum. 

\section{Chemical Baum-Welch Algorithm}\label{sec::reac}
\subsection{Reaction Networks}
Following~\cite{virinchi2018reaction}, we recall some concepts from reaction network theory
\cite{feinberg72chemical,horn72necessary,Fein79,Manoj_2011Catalysis,anderson2010product,virinchi2017stochastic}.

Fix a finite set $S$ of species. An $S$-reaction, or simply a \textbf{reaction} when $S$ is understood from context, is a formal chemical equation
\[
  \sum_{X\in S} y_X X \rightarrow \sum_{X\in S}y'_XX
\]
where the numbers $y_X,y'_X\in\mathbb{Z}_{\geq 0}$ are the \textbf{stoichiometric coefficients} of species $X$ on the \textbf{reactant} side and \textbf{product} side respectively. A \textbf{reaction network} is a pair $(S,R)$ where $R$ is a finite set of $S$-reactions. A \textbf{reaction system} is a triple $(\S,\R,k)$ where $(\S,\R)$ is a reaction network and $k:\R\to\mathbb{R}_{> 0}$ is called the \textbf{rate function}.

As is common when specifying reaction networks, we will find it convenient to explicitly specify only a set of chemical equations, leaving the set of species to be inferred by the reader.

Fix a reaction system $(\S,\R,k)$. \textbf{Deterministic Mass Action Kinetics} describes a system of ordinary differential equations on the concentration variables $\{x_i(t)\mid i\in \S\}$ according to:
\[\dot x_i(t)=\sum_{a\to b\in\R}k_{a\to b}(b_i-a_i)\prod_{j\in\S}x_j(t)^{a_j}\] 

\subsection{Baum-Welch Reaction Network}\label{sec:BWRN}
In this section we will describe a reaction networks for each part of the Baum-Welch algorithm.

Fix an HMM $\mathcal{M}=(H,V,\theta,\psi,\pi)$. Pick an arbitrary hidden state $h^*\in H$ and an arbitrary visible state $v^*\in V$. Picking these states $h^* \in H$ and $v^*\in V$ is merely an artifice to break symmetry akin to selecting leaders, and our results hold independent of these choices. Also fix a length $L\in \mathbb{Z}_{>0}$ of observed sequence.

We first work out in full detail how the forward algorithm of the Baum-Welch algorithm may be translated into chemical reactions. Suppose a length $L$ sequence $(v_1,v_2,\dots,v_L)\in V^L$ of visible states is observed. Then recall that the forward algorithm uses the following recursion:
\begin{itemize}
\item Initialisation: $\alpha_{1h} = \pi_h\psi_{hv_1}$ for all $h\in H$,
\item Recursion: $\alpha_{l h}=\sum_{g\in H} \alpha_{l-1,g}\theta_{gh}\psi_{h v_l}$ for all $h\in H$ and $l=2,\ldots,L$.
\end{itemize}

Notice this implies
\begin{itemize}
\item $\alpha_{1h}\times \pi_{h^*}\psi_{h^*v_1} = \alpha_{1h^*}\times\pi_h\psi_{hv_1}$ for all $h\in H\setminus\{h^*\}$,
\item $\alpha_{l h}\times \left(\sum_{g\in H} \alpha_{l-1,g}\theta_{gh^*}\psi_{h^* v_l}\right)=\alpha_{l h^*}\times\left(\sum_{g\in H} \alpha_{l-1,g}\theta_{gh}\psi_{h v_l}\right)$ for all $h\in H\setminus\{h^*\}$ and $l=2,\ldots,L$.
\end{itemize}

This prompts the use of the following reactions for the initialization step:
\begin{equation*}
  \begin{gathered}
    \alpha_{1h}+\pi_{h^*}+\psi_{h^*v_1}\ce{->[]} \alpha_{1h^*}+\pi_{h^*}+\psi_{h^*v_1}\\
    \alpha_{1h^*}+\pi_{h}+\psi_{hv_1}\ce{->[]} \alpha_{1h}+\pi_{h}+\psi_{hv_1}
  \end{gathered}
\end{equation*}
for all $h\in H\setminus\{h^*\}$ and $w\in V$. By design $\alpha_{1h}\times \pi_{h^*}\psi_{h^*v_1} = \alpha_{1h^*}\times\pi_h\psi_{hv_1}$ is the balance equation for the pair of reactions corresponding to each $h\in H\setminus\{h^*\}$.

Similarly for the recursion step, we use the following reactions:
\begin{equation*}
  \begin{gathered}
    \alpha_{l h}+\alpha_{l-1,g}+ \theta_{gh^*}+\psi_{h^*v_l}\ce{->[]} \alpha_{l h^*}+\alpha_{l-1,g}+\theta_{gh^*}+\psi_{h^*v_l}\\
    \alpha_{l h^*}+\alpha_{l-1,g}+\theta_{gh}+\psi_{hv_l}\ce{->[]} \alpha_{l h}+\alpha_{l-1,g}+ \theta_{gh}+\psi_{hv_l}
  \end{gathered}
\end{equation*}
for all $g\in H$, $h\in H\setminus\{h^*\}, l=2,\ldots, L$ and $w\in V$. Again by design \[\alpha_{l h}\times \left(\sum_{g\in H} \alpha_{l-1,g}\theta_{gh^*}\psi_{h^* v_l}\right)=\alpha_{l h^*}\times\left(\sum_{g\in H} \alpha_{l-1,g}\theta_{gh}\psi_{h v_l}\right)\] is the balance equation for the set of reactions corresponding to each $(h,l)\in H\setminus\{h^*\}\times\{2,\ldots,L\}$.

The above reactions depend on the observed sequence $(v_1,v_2,\dots,v_L)\in V^L$ of visible state. And this is a problem because one would have to design different reaction networks for different observed sequences. To solve this problem we introduce the species $E_{lw}$ with $l=1,\dots,L$ and $w\in V$.
Now with the $E_{lw}$ species, we use the following reactions for the forward algorithm:
\begin{equation*}
  \begin{gathered}
    \alpha_{1h}+\pi_{h^*}+\psi_{h^*w}+E_{1w}\ce{->[]} \alpha_{1h^*}+\pi_{h^*}+\psi_{h^*w}+E_{1w}\\
    \alpha_{1h^*}+\pi_{h}+\psi_{hw}+E_{1w}\ce{->[]} \alpha_{1h}+\pi_{h}+\psi_{hw}+E_{1w}
  \end{gathered}
\end{equation*}
for all $h\in H\setminus\{h^*\}$ and $w\in V$.
\begin{equation*}  
  \begin{gathered}
    \alpha_{l h}+\alpha_{l-1,g}+ \theta_{gh^*}+\psi_{h^*w}+E_{lw}\ce{->[]} \alpha_{l h^*}+\alpha_{l-1,g}+\theta_{gh^*}+\psi_{h^*w}+E_{lw}\\
    \alpha_{l h^*}+\alpha_{l-1,g}+\theta_{gh}+\psi_{hw}+E_{lw}\ce{->[]} \alpha_{l h}+\alpha_{l-1,g}+ \theta_{gh}+\psi_{hw}+E_{lw}
  \end{gathered}
\end{equation*}
for  all $g\in H$, $h\in H\setminus\{h^*\}, l=2,\ldots, L$ and $w\in V$. The $E_{lw}$ species are to be initialized such that $E_{lw}=1$ iff $w=v_l$. So different observed sequences can now be processed by the same reaction network, by appropriately initializing the species $E_{lw}$.

The other parts of the Baum-Welch algorithm may be translated into chemical reactions using a similar logic. We call the resulting reaction network as the \textbf{Baum-Welch Reaction Network} $BW(\mathcal{M},h^*,v^*,L)$. It consists of four subnetworks corresponding to the four parts of the Baum-Welch algorithm, as shown in Table~\ref{tab:BWCRN}.

\begin{longtable}[!h]{|>{\centering\arraybackslash}m{0.35\linewidth}|>{\centering\arraybackslash}m{0.65\linewidth}|}
  \caption{\textbf{Baum-Welch Reaction Network}: The steps and reactions are for all $g,h\in H, w\in V$ and $l=1,\dots,L-1$. Notice there are some null reactions of the form $\sum a_iX_i\to \sum a_iX_i$. As these null reaction have no effect on the dynamics, they can be ignored.}
  \label{tab:BWCRN}\\
  \hline
  \textbf{Baum-Welch Algorithm} & \textbf{Baum-Welch Reaction Network}\\
  \hline
  \vbox{\[\alpha_{1h} = \pi_h\psi_{hv_1}\]
  }
  &\vbox{
    \begin{equation*}
      \begin{aligned}
        \alpha_{1h}+\pi_{h^*}+\psi_{h^*w}+E_{1w}&\ce{->[]} \alpha_{1h^*}+\pi_{h^*}+\psi_{h^*w}+E_{1w}\\
        \alpha_{1h^*}+\pi_{h}+\psi_{hw}+E_{1w}&\ce{->[]} \alpha_{1h}+\pi_{h}+\psi_{hw}+E_{1w}
      \end{aligned}
    \end{equation*}
  }\\
  \vbox{\[\alpha_{l+1, h}=\sum_{g\in H} \alpha_{lg}\theta_{gh}\psi_{h v_{l+1}}\]
  }
  &\vbox{
    \begin{equation*}
      \begin{gathered}
        \alpha_{l+1, h}+\alpha_{lg}+ \theta_{gh^*}+\psi_{h^*w}+E_{l+1,w}\ce{->[]}\\\alpha_{l+1, h^*}+\alpha_{lg}+\theta_{gh^*}+\psi_{h^*w}+E_{l+1,w}\\
        \alpha_{l+1, h^*}+\alpha_{lg}+\theta_{gh}+\psi_{hw}+E_{l+1,w}\ce{->[]}\\\alpha_{l+1, h}+\alpha_{lg}+ \theta_{gh}+\psi_{hw}+E_{l+1,w}
      \end{gathered}
    \end{equation*}
  }\\
  \hline
  \vbox{\[\beta_{L h}=1\]
  }
  & 
  \multirow{2}{*}{\vbox{
      \begin{equation*}
        \begin{gathered}
          \beta_{l h}+ \beta_{l+1,g}+ \theta_{h^*g} + \psi_{gw}+E_{l+1,w}\ce{->[]}\\\beta_{l h^*}+ \beta_{l+1,g}+\theta_{h^*g} + \psi_{gw}+E_{l+1,w}\\
          \beta_{l h^*} + \beta_{l+1,g}+ \theta_{hg} + \psi_{gw}+E_{l+1,w}\ce{->[]}\\\beta_{l h}+ \beta_{l+1,g}+ \theta_{hg} + \psi_{gw}+E_{l+1,w}
        \end{gathered}
      \end{equation*}
    }
  }
  \\
  \vbox{\[\beta_{l h}=\sum_{g\in H} \theta_{hg}\psi_{g v_{l+1}}\beta_{l+1,g}\]
  }
  & \\
  \hline
  \vbox{\[
  \gamma_{lh}=\frac{\alpha_{l h}\beta_{l h}}{\sum_{g\in H} \alpha_{l g}\beta_{l g}}
\]
}
  &\vbox{
    \begin{equation*}
      \begin{gathered}
        \gamma_{lh}+\alpha_{l h^*}+\beta_{l h^*}\ce{->[]}\gamma_{lh^*}+\alpha_{l h^*}+\beta_{l h^*}\\
        \gamma_{lh^*}+\alpha_{l h}+\beta_{l h}\ce{->[]}\gamma_{lh}+\alpha_{l h}+\beta_{l h}
      \end{gathered}
    \end{equation*}
  }\\
  \vbox{\[
  \xi_{lgh}=\frac{\alpha_{l g}\theta_{gh}\psi_{ hv_{l+1}}\beta_{l+1,h}}{\sum_{f\in H} \alpha_{l f}\beta_{l f}}
\]
}
  &\vbox{
    \begin{align*}
      \begin{split}	
        \xi_{lgh}+&\alpha_{l h^*}+\theta_{h^*h^*}+\beta_{l+1, h^*}+\psi_{h^* w}+E_{l+1,w}\ce{->[]}\
        \\\xi_{lh^*h^*}+&\alpha_{l h^*}+\theta_{h^*h^*}+\beta_{l+1, h^*}+\psi_{h^* w}+E_{l+1,w}\
        \\\xi_{lh^*h^*}+&\alpha_{l g}+\theta_{gh}+\beta_{l+1, g}+\psi_{h w}+E_{l+1,w}\ce{->[]}\
        \\\xi_{lgh}+&\alpha_{l g}+\theta_{gh}+\beta_{l+1, g}+\psi_{h w}+E_{l+1,w}
      \end{split}
    \end{align*}		
  }\\    
  \hline
  \vbox{\[\theta_{gh}\leftarrow\frac{\sum_{l=1}^{L-1}\xi_{lgh}}{\sum_{l=1}^{L-1}\sum_{f\in H}\xi_{lgf}}\]
  }
  &\vbox{
    \begin{equation*}
      \begin{gathered}
        \theta_{gh}+\xi_{lgh^*}\ce{->[]}\theta_{gh^*}+\xi_{lgh^*}\\
        \theta_{gh^*}+\xi_{lgh}\ce{->[]}\theta_{gh}+\xi_{lgh}
      \end{gathered}
    \end{equation*}
  }\\
  \vbox{\[\psi_{hw}\leftarrow\frac{\sum_{l=1}^L\gamma_{lh}\delta_{w,v_l}}{\sum_{l=1}^L\gamma_{lh}}\]
  }
  &\vbox{
    \begin{equation*}
      \begin{gathered}
        \psi_{hw}+\gamma_{lh}+E_{lv^*}\ce{->[]}\psi_{hv^*}+\gamma_{lh}+E_{lv^*}\\
        \psi_{hv^*}+\gamma_{lh}+E_{lw}\ce{->[]}\psi_{hw}+\gamma_{lh}+E_{lw}
      \end{gathered}
    \end{equation*}
  }\\
  \hline
\end{longtable}

The Baum-Welch reaction network described above has a special structure. Every reaction is a monomolecular transformation catalyzed by a set of species. The reverse transformation is also present, catalyzed by a different set of species to give the network a ``futile cycle''~\cite{tu2006metabolic} structure. In addition, each connected component in the undirected graph representation of the network has a topology with all transformations happening to and from a central species. This prompts the following definitions.
\begin{figure}[h!]
  \centering
  \hspace{4cm}\scalebox{0.8}{\subfloat{
  		\centering
  		{\def
        \svgscale{0.7}
        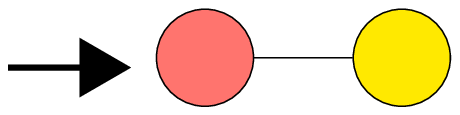 }}}\
  \\ \vspace{15pt}	\scalebox{0.8}{\subfloat{{\def
        \svgwidth{125pt}
\begingroup%
  \makeatletter%
  \providecommand\color[2][]{%
    \errmessage{(Inkscape) Color is used for the text in Inkscape, but the package 'color.sty' is not loaded}%
    \renewcommand\color[2][]{}%
  }%
  \providecommand\transparent[1]{%
    \errmessage{(Inkscape) Transparency is used (non-zero) for the text in Inkscape, but the package 'transparent.sty' is not loaded}%
    \renewcommand\transparent[1]{}%
  }%
  \providecommand\rotatebox[2]{#2}%
  \newcommand*\fsize{\dimexpr\f@size pt\relax}%
  \newcommand*\lineheight[1]{\fontsize{\fsize}{#1\fsize}\selectfont}%
  \ifx\svgwidth\undefined%
    \setlength{\unitlength}{183.00106259bp}%
    \ifx\svgscale\undefined%
      \relax%
    \else%
      \setlength{\unitlength}{\unitlength * \real{\svgscale}}%
    \fi%
  \else%
    \setlength{\unitlength}{\svgwidth}%
  \fi%
  \global\let\svgwidth\undefined%
  \global\let\svgscale\undefined%
  \makeatother%
  \begin{picture}(1,0.97170113)%
    \lineheight{1}%
    \setlength\tabcolsep{0pt}%
    \put(0,0){\includegraphics[width=\unitlength]{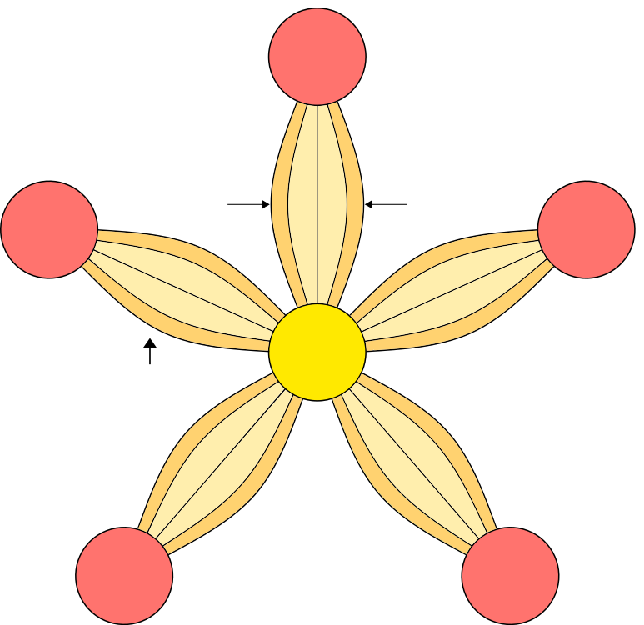}}%
    \put(0.58565819,0.71328212){\color[rgb]{0,0,0}\makebox(0,0)[lt]{\lineheight{1.25}\smash{\begin{tabular}[t]{l}$|H|\times |V|$ \end{tabular}}}}%
    \put(0.43754411,0.41716127){\color[rgb]{0,0,0}\makebox(0,0)[lt]{\lineheight{1.25}\smash{\begin{tabular}[t]{l}$\alpha_{lh^*}$\end{tabular}}}}%
    \put(0.87713189,0.61132016){\color[rgb]{0,0,0}\makebox(0,0)[lt]{\lineheight{1.25}\smash{\begin{tabular}[t]{l}$\alpha_{lh}$\end{tabular}}}}%
    \put(0.74803028,0.40534883){\color[rgb]{0,0,0}\makebox(0,0)[lt]{\lineheight{1.25}\smash{\begin{tabular}[t]{l}$k^{\alpha}_{lh}$\end{tabular}}}}%
    \put(0.16342096,0.36373626){\color[rgb]{0,0,0}\makebox(0,0)[lt]{\lineheight{1.25}\smash{\begin{tabular}[t]{l}Petal\end{tabular}}}}%
    \put(0.40624896,0.06558023){\color[rgb]{0,0,0}\makebox(0,0)[lt]{\lineheight{1.25}\smash{\begin{tabular}[t]{l}Flower\end{tabular}}}}%
  \end{picture}%
\endgroup%
 }
      {\def
        \svgwidth{125pt}
\begingroup%
  \makeatletter%
  \providecommand\color[2][]{%
    \errmessage{(Inkscape) Color is used for the text in Inkscape, but the package 'color.sty' is not loaded}%
    \renewcommand\color[2][]{}%
  }%
  \providecommand\transparent[1]{%
    \errmessage{(Inkscape) Transparency is used (non-zero) for the text in Inkscape, but the package 'transparent.sty' is not loaded}%
    \renewcommand\transparent[1]{}%
  }%
  \providecommand\rotatebox[2]{#2}%
  \newcommand*\fsize{\dimexpr\f@size pt\relax}%
  \newcommand*\lineheight[1]{\fontsize{\fsize}{#1\fsize}\selectfont}%
  \ifx\svgwidth\undefined%
    \setlength{\unitlength}{183.00106259bp}%
    \ifx\svgscale\undefined%
      \relax%
    \else%
      \setlength{\unitlength}{\unitlength * \real{\svgscale}}%
    \fi%
  \else%
    \setlength{\unitlength}{\svgwidth}%
  \fi%
  \global\let\svgwidth\undefined%
  \global\let\svgscale\undefined%
  \makeatother%
  \begin{picture}(1,0.97170113)%
    \lineheight{1}%
    \setlength\tabcolsep{0pt}%
    \put(0,0){\includegraphics[width=\unitlength]{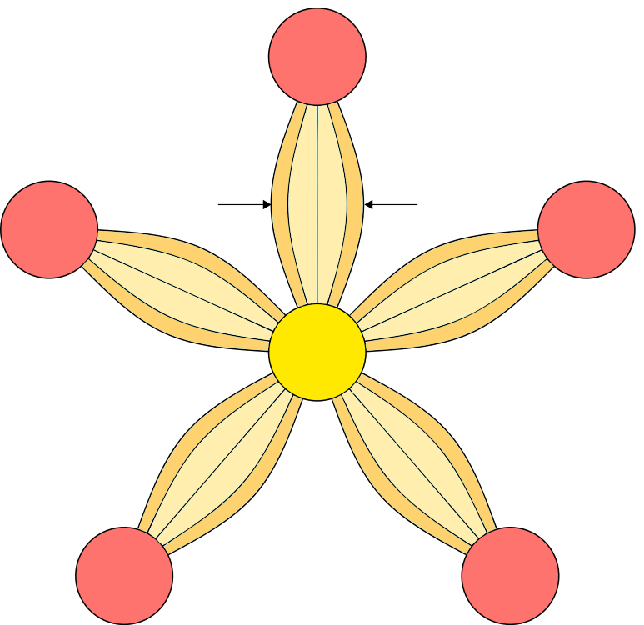}}%
    \put(0.58565819,0.71328543){\color[rgb]{0,0,0}\makebox(0,0)[lt]{\lineheight{1.25}\smash{\begin{tabular}[t]{l}$|H|\times |V|$ \end{tabular}}}}%
    \put(0.44420472,0.41716127){\color[rgb]{0,0,0}\makebox(0,0)[lt]{\lineheight{1.25}\smash{\begin{tabular}[t]{l}$\beta_{lh^*}$\end{tabular}}}}%
    \put(0.87713756,0.61131638){\color[rgb]{0,0,0}\makebox(0,0)[lt]{\lineheight{1.25}\smash{\begin{tabular}[t]{l}$\beta_{lh}$\end{tabular}}}}%
    \put(0.74803028,0.40534883){\color[rgb]{0,0,0}\makebox(0,0)[lt]{\lineheight{1.25}\smash{\begin{tabular}[t]{l}$k^{\beta}_{lh}$\end{tabular}}}}%
  \end{picture}%
\endgroup%
 }
      {\def
        \svgwidth{125pt}
\begingroup%
  \makeatletter%
  \providecommand\color[2][]{%
    \errmessage{(Inkscape) Color is used for the text in Inkscape, but the package 'color.sty' is not loaded}%
    \renewcommand\color[2][]{}%
  }%
  \providecommand\transparent[1]{%
    \errmessage{(Inkscape) Transparency is used (non-zero) for the text in Inkscape, but the package 'transparent.sty' is not loaded}%
    \renewcommand\transparent[1]{}%
  }%
  \providecommand\rotatebox[2]{#2}%
  \newcommand*\fsize{\dimexpr\f@size pt\relax}%
  \newcommand*\lineheight[1]{\fontsize{\fsize}{#1\fsize}\selectfont}%
  \ifx\svgwidth\undefined%
    \setlength{\unitlength}{183.00106259bp}%
    \ifx\svgscale\undefined%
      \relax%
    \else%
      \setlength{\unitlength}{\unitlength * \real{\svgscale}}%
    \fi%
  \else%
    \setlength{\unitlength}{\svgwidth}%
  \fi%
  \global\let\svgwidth\undefined%
  \global\let\svgscale\undefined%
  \makeatother%
  \begin{picture}(1,0.97170113)%
    \lineheight{1}%
    \setlength\tabcolsep{0pt}%
    \put(0,0){\includegraphics[width=\unitlength]{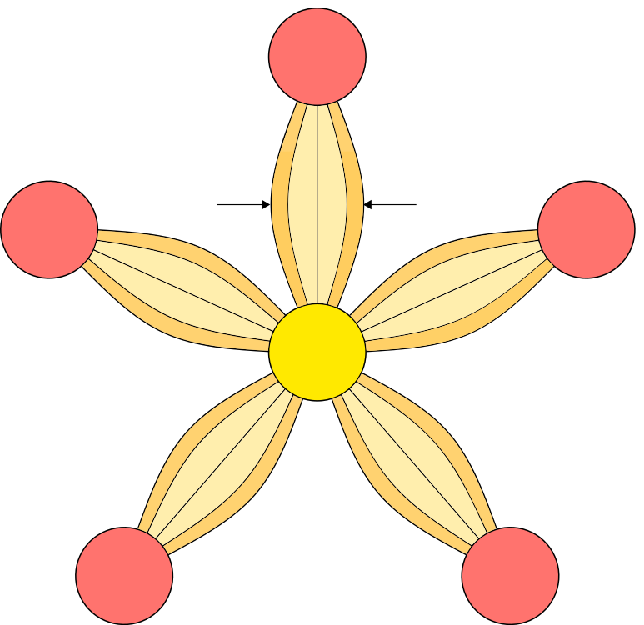}}%
    \put(0.58565098,0.71328543){\color[rgb]{0,0,0}\makebox(0,0)[lt]{\lineheight{1.25}\smash{\begin{tabular}[t]{l}$|V|$ \end{tabular}}}}%
    \put(0.4352143,0.41716127){\color[rgb]{0,0,0}\makebox(0,0)[lt]{\lineheight{1.25}\smash{\begin{tabular}[t]{l}$\xi_{l\scaleto{h^*h^*}{3pt}}$\end{tabular}}}}%
    \put(0.86791072,0.61132016){\color[rgb]{0,0,0}\makebox(0,0)[lt]{\lineheight{1.25}\smash{\begin{tabular}[t]{l}$\xi_{lgh}$\end{tabular}}}}%
    \put(0.74803643,0.40535237){\color[rgb]{0,0,0}\makebox(0,0)[lt]{\lineheight{1.25}\smash{\begin{tabular}[t]{l}$k^{\xi}_{lgh}$\end{tabular}}}}%
  \end{picture}%
\endgroup%
 }}
  }\\
  \scalebox{0.8}{\subfloat{
      {\def
        \svgwidth{125pt}
\begingroup%
  \makeatletter%
  \providecommand\color[2][]{%
    \errmessage{(Inkscape) Color is used for the text in Inkscape, but the package 'color.sty' is not loaded}%
    \renewcommand\color[2][]{}%
  }%
  \providecommand\transparent[1]{%
    \errmessage{(Inkscape) Transparency is used (non-zero) for the text in Inkscape, but the package 'transparent.sty' is not loaded}%
    \renewcommand\transparent[1]{}%
  }%
  \providecommand\rotatebox[2]{#2}%
  \newcommand*\fsize{\dimexpr\f@size pt\relax}%
  \newcommand*\lineheight[1]{\fontsize{\fsize}{#1\fsize}\selectfont}%
  \ifx\svgwidth\undefined%
    \setlength{\unitlength}{183.00106259bp}%
    \ifx\svgscale\undefined%
      \relax%
    \else%
      \setlength{\unitlength}{\unitlength * \real{\svgscale}}%
    \fi%
  \else%
    \setlength{\unitlength}{\svgwidth}%
  \fi%
  \global\let\svgwidth\undefined%
  \global\let\svgscale\undefined%
  \makeatother%
  \begin{picture}(1,0.97170113)%
    \lineheight{1}%
    \setlength\tabcolsep{0pt}%
    \put(0,0){\includegraphics[width=\unitlength]{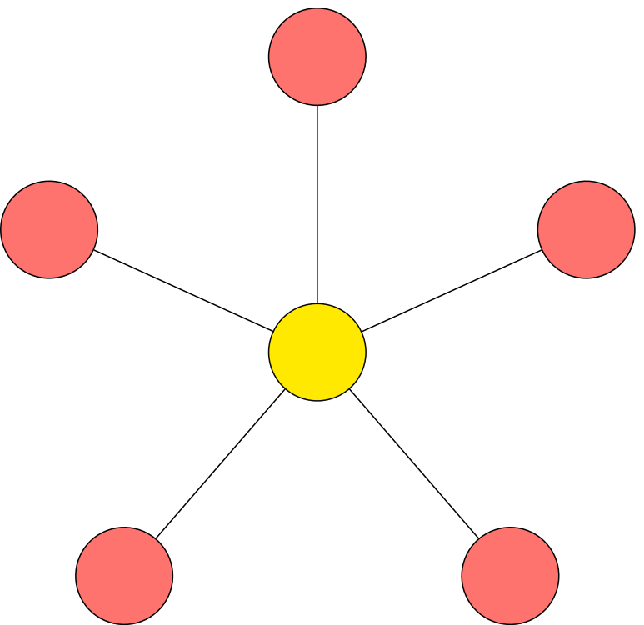}}%
    \put(0.45378739,0.41716127){\color[rgb]{0,0,0}\makebox(0,0)[lt]{\lineheight{1.25}\smash{\begin{tabular}[t]{l}$\gamma_{l\scaleto{h^*}{4pt}}$\end{tabular}}}}%
    \put(0.87713189,0.61132016){\color[rgb]{0,0,0}\makebox(0,0)[lt]{\lineheight{1.25}\smash{\begin{tabular}[t]{l}$\gamma_{lh}$\end{tabular}}}}%
    \put(0.74803028,0.45181816){\color[rgb]{0,0,0}\makebox(0,0)[lt]{\lineheight{1.25}\smash{\begin{tabular}[t]{l}$k^{\gamma}_{lh}$\end{tabular}}}}%
  \end{picture}%
\endgroup%
 }
      {\def
        \svgwidth{125pt}
\begingroup%
  \makeatletter%
  \providecommand\color[2][]{%
    \errmessage{(Inkscape) Color is used for the text in Inkscape, but the package 'color.sty' is not loaded}%
    \renewcommand\color[2][]{}%
  }%
  \providecommand\transparent[1]{%
    \errmessage{(Inkscape) Transparency is used (non-zero) for the text in Inkscape, but the package 'transparent.sty' is not loaded}%
    \renewcommand\transparent[1]{}%
  }%
  \providecommand\rotatebox[2]{#2}%
  \newcommand*\fsize{\dimexpr\f@size pt\relax}%
  \newcommand*\lineheight[1]{\fontsize{\fsize}{#1\fsize}\selectfont}%
  \ifx\svgwidth\undefined%
    \setlength{\unitlength}{183.00106259bp}%
    \ifx\svgscale\undefined%
      \relax%
    \else%
      \setlength{\unitlength}{\unitlength * \real{\svgscale}}%
    \fi%
  \else%
    \setlength{\unitlength}{\svgwidth}%
  \fi%
  \global\let\svgwidth\undefined%
  \global\let\svgscale\undefined%
  \makeatother%
  \begin{picture}(1,0.97170113)%
    \lineheight{1}%
    \setlength\tabcolsep{0pt}%
    \put(0,0){\includegraphics[width=\unitlength]{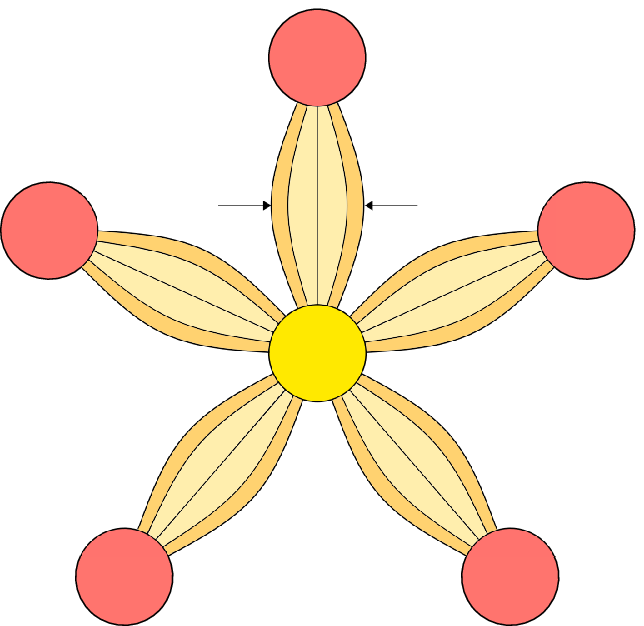}}%
    \put(0.58565819,0.71328543){\color[rgb]{0,0,0}\makebox(0,0)[lt]{\lineheight{1.25}\smash{\begin{tabular}[t]{l}$L$ \end{tabular}}}}%
    \put(0.43703264,0.41716127){\color[rgb]{0,0,0}\makebox(0,0)[lt]{\lineheight{1.25}\smash{\begin{tabular}[t]{l}$\theta_{gh^*}$\end{tabular}}}}%
    \put(0.87713756,0.61131638){\color[rgb]{0,0,0}\makebox(0,0)[lt]{\lineheight{1.25}\smash{\begin{tabular}[t]{l}$\theta_{gh}$\end{tabular}}}}%
    \put(0.74803028,0.40534883){\color[rgb]{0,0,0}\makebox(0,0)[lt]{\lineheight{1.25}\smash{\begin{tabular}[t]{l}$k^{\theta}_{gh}$\end{tabular}}}}%
  \end{picture}%
\endgroup%
 }
      {\def
        \svgwidth{125pt}
\begingroup%
  \makeatletter%
  \providecommand\color[2][]{%
    \errmessage{(Inkscape) Color is used for the text in Inkscape, but the package 'color.sty' is not loaded}%
    \renewcommand\color[2][]{}%
  }%
  \providecommand\transparent[1]{%
    \errmessage{(Inkscape) Transparency is used (non-zero) for the text in Inkscape, but the package 'transparent.sty' is not loaded}%
    \renewcommand\transparent[1]{}%
  }%
  \providecommand\rotatebox[2]{#2}%
  \newcommand*\fsize{\dimexpr\f@size pt\relax}%
  \newcommand*\lineheight[1]{\fontsize{\fsize}{#1\fsize}\selectfont}%
  \ifx\svgwidth\undefined%
    \setlength{\unitlength}{183.00106259bp}%
    \ifx\svgscale\undefined%
      \relax%
    \else%
      \setlength{\unitlength}{\unitlength * \real{\svgscale}}%
    \fi%
  \else%
    \setlength{\unitlength}{\svgwidth}%
  \fi%
  \global\let\svgwidth\undefined%
  \global\let\svgscale\undefined%
  \makeatother%
  \begin{picture}(1,0.97170113)%
    \lineheight{1}%
    \setlength\tabcolsep{0pt}%
    \put(0,0){\includegraphics[width=\unitlength]{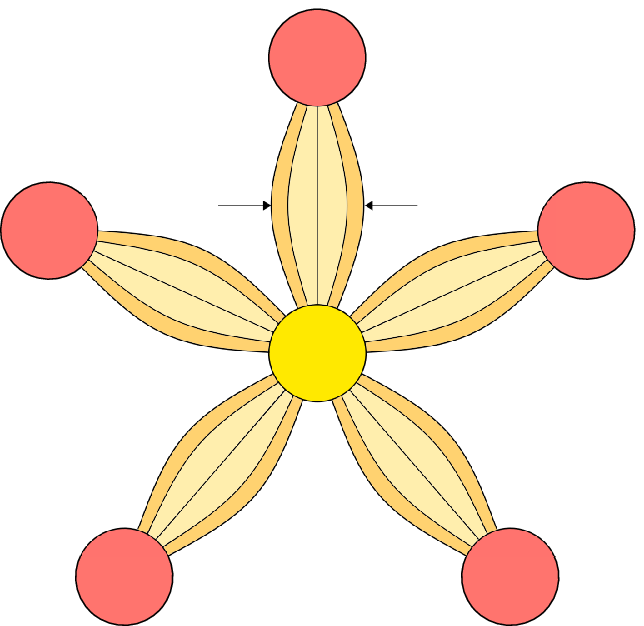}}%
    \put(0.58565819,0.71328543){\color[rgb]{0,0,0}\makebox(0,0)[lt]{\lineheight{1.25}\smash{\begin{tabular}[t]{l}$L$ \end{tabular}}}}%
    \put(0.42934825,0.41716127){\color[rgb]{0,0,0}\makebox(0,0)[lt]{\lineheight{1.25}\smash{\begin{tabular}[t]{l}$\psi_{hv^*}$\end{tabular}}}}%
    \put(0.863818,0.61131638){\color[rgb]{0,0,0}\makebox(0,0)[lt]{\lineheight{1.25}\smash{\begin{tabular}[t]{l}$\psi_{hw}$\end{tabular}}}}%
    \put(0.74803028,0.40534883){\color[rgb]{0,0,0}\makebox(0,0)[lt]{\lineheight{1.25}\smash{\begin{tabular}[t]{l}$k^{\psi}_{hw}$\end{tabular}}}}%
  \end{picture}%
\endgroup%
 }}}\\		
  \caption{\textbf{The Baum-Welch Reaction Network represented as an undirected graph.} (a) Each reaction is a unimolecular transformation driven by catalysts. The nodes represents the species undergoing transformation. The edge represents a pair of reactions which drive this transformation backwards and forwards in a futile cycle. Catalysts are omitted for clarity. (b) The network decomposes into a collection of disjoint \textbf{flowers}. Nodes represent species and edges represent pairs of reactions, species $\alpha,\gamma$ have $L$ flowers each, $\beta,\xi$ have $L-1$ flowers each, and species $\theta$ and $\psi$ have $|H|$ flowers each (not shown in figure). All reactions in the same petal have the same specific rate constant, so the dimension of the space of permissible rate constants is equal to the number of petals in the graph.}
  \label{fig::BaumWelchRN}
\end{figure}
\begin{definition}[Flowers, petals, gardens]
  A \textbf{graph} is a triple $(\text{Nodes}, \text{Edges}, \pi)$ where Nodes and Edges are finite sets and $\pi$ is a map from Edges to unordered pairs of elements from Nodes. A \textbf{flower} is a graph with a special node $n^*$ such that for every edge $e\in \text{Edges}$ we have $n^*\in\pi(e)$. A \textbf{garden} is a graph which is a union of disjoint flowers. A \textbf{petal} is a set of all edges $e$ which have the same $\pi(e)$, i.e. they are incident between the same pair of nodes.
\end{definition}

Figure~\ref{fig::BaumWelchRN} shows how the Baum-Welch reaction network can be represented as a garden graph, in which species are nodes and reactions are edges. 

A collection of specific rates is \textbf{permissible} if all reactions in a petal have the same rate. However, different petals may have different rates. We will denote the specific rate for a petal by superscripting the type of the species and subscripting its indices. For example, the specific rate  for reactions in the petal for $\alpha_{lh}$ would be denoted as $k^\alpha_{lh}$. The notation for the remaining rate constants can be read from Figure~ \ref{fig::BaumWelchRN}.
\begin{remark}\label{rmk:topology}
  The ``flower'' topology we have employed for the Baum-Welch reaction network is only one among several possibilities. The important thing is to achieve connectivity between different nodes that ought to be connected, ensuring that the ratio of the concentrations of the species denoted by adjacent nodes takes the value as prescribed by the Baum-Welch algorithm. Therefore, many other connection topologies can be imagined, for example a ring topology where each node is connected to two other nodes while maintaining the same connected components. In fact, there are obvious disadvantages to the star topology, with a single point of failure, whereas the ring topology appears more resilient. How network topology affects basins of attraction, rates of convergence, and emergence of spurious equilibria in our algorithm is a compelling question beyond the scope of this present work.
\end{remark}

The Baum-Welch reaction network with a permissible choice of rate constants $k$ will define a \textbf{Baum-Welch reaction system $(BW(\mathcal{M},h^*,v^*,L),k)$} whose deterministic mass action kinetics equations will perform a continuous time version of the Baum-Welch algorithm. We call this the \textbf{Chemical Baum-Welch Algorithm}, and describe it in Algorithm~\ref{SimHMMCRN}.\\
\\
\RestyleAlgo{boxruled}
\begin{algorithm}[H]
  \SetAlgoLined
  \KwIn{An HMM $\mathcal{M}=(H,V,\theta,\psi,\pi)$ and an observed sequence $v\in V^L$}
  \KwOut{Parameters $\hat\theta\in \mathbb{R}_{\geq0}^{|H|\times |H|},\hat\psi\in \mathbb{R}_{\geq0}^{|H|\times |V|}$}
  \textbf{Initialization of concentrations at t=0:}
  \begin{enumerate}
  \item For $w\in V$ and $l=1,\ldots,L$, initialize $E_{lw}(0)$ such that $E_{lw}(0)=	\begin{cases}
      1\text{ if }w=v_l\
      \\0\text{ otherwise}
    \end{cases}$
  \item For $g\in H$, initialize $\beta_{Lg}=\beta$
  \item For every other species, initialize its concentration arbitrarily in $\mathbb{R}_{>0}$.
  \end{enumerate}
  \textbf{Algorithm:}\\
  \hspace{0.5cm}Run the Baum-Welch reaction system with mass action kinetics until convergence.
  \caption{Chemical Baum-Welch Algorithm}
  \label{SimHMMCRN}
\end{algorithm}

\section{Analysis and Simulations}
\label{sec:analysis}

The Baum-Welch reaction network has number of species of each type as follows:
\[
  \begin{array}{|c||c|c|c|c|c|c|c|c|}
    \hline\text{Type} & \pi & \alpha & \beta & \theta & \psi & E & \gamma & \xi \
    \\	\hline\text{Number} & \,|H|\, & \,|H|L \,& \,|H|L \,& \,|H|^2 \,&\, |H||V|\, 
                                                             & \,L|V|\, & \,|H|L\, & \,|H|^2(L-1)\
    \\\hline
  \end{array}
\]
The total number of species is $|H| + 3|H|L + |H|^2 + |H||V| + L|V| + |H|^2(L-1)$ which is $O(L|H|^2 +|H||V| + L|V|)$.
The number of reactions (ignoring the null reactions of the form $\sum a_iX_i\to \sum a_iX_i$) in each part is:
\[
  \begin{array}{|c|c|c|c|}
    \hline\text{Forward} &\text{Backward} &\text{Expectation} &\text{Maximization}\
    \\\hline 2(|H|-1)V +  & (2|H|(|H|-1)|V|)
                                          & 2L(|H|-1) +  &2|H|(|H|-1)(L-1) \
    \\ 2|H|(|H|-1)(L-1)|V| & (L-1)  &2(L-1)(|H|^2-1) &+ 2|H|L(|V|-1)\
    \\\hline
  \end{array}
\]
so that the total number of reactions is $O(|H|^2L|V|)$.

The first theorem shows that the Chemical Baum-Welch Algorithm recovers all of the Baum-Welch equilibria.

\begin{theorem}\label{thm:bw}
  Every fixed point of the Baum-Welch algorithm for an HMM $\mathcal{M}=(H,V,\theta,\psi,\pi)$ is a fixed point for the corresponding Chemical Baum-Welch Algorithm with permissible rates $k$.
\end{theorem}
See \ref{pf:bw} for proof.

We say a vector of real numbers is \textbf{positive} if all its components are strictly greater than $0$. The next theorem shows that positive equilibria of the Chemical Baum-Welch Algorithm are also fixed points of the Baum-Welch algorithm.

\begin{theorem}\label{thm:bwcrn}
  Every \text{\textbf{positive}} fixed point for the Chemical Baum-Welch Algorithm on a {Baum-Welch Reaction system} $(BW(\mathcal{M},h^*,v^*,L),k)$ with permissible rate $k$ is a fixed point for the Baum-Welch algorithm for the HMM $\mathcal{M}=(H,V,\theta,\psi,\pi)$.
\end{theorem}
See \ref{pf:bwcrn} for proof.

The Baum-Welch algorithm is an iterative algorithm. The E step and the M step are run iteratively. In contrast, the Chemical Baum-Welch algorithm is a generalized EM algorithm~\cite{mclachlan2007algorithm} where all the reactions are run at the same time in a single-pot reaction.

The next theorem shows that the $E$ step consisting of the forward network, the backward network, and the $E$ network, converges exponentially fast to the correct equilibrium if the $\theta$ and $\psi$ species are held fixed at a {\em positive} point. 

\begin{theorem}\label{th:EMexpconv1}
  For the {Baum-Welch Reaction System} $(BW(\mathcal{M},h^*,v^*,L),k)$ with permissible rates $k$, if the concentrations of $\theta$ and $\psi$ species are held fixed at a {\em positive} point then the Forward, Backward and Expection step reaction systems on $\alpha,\beta,\gamma$ and $\psi$ species converge to equilibrium exponentially fast.
\end{theorem}
See \ref{pf:EMexpconv1} for proof.

The next theorem shows that if the $\alpha,\beta,\gamma,\xi$ species are held fixed at a {\em positive} point, then the $M$ step consisting of reactions modifying the $\theta$ and $\psi$ species converges exponentially fast to the correct equilibrium.
\begin{theorem}\label{th:EMexpconv2}
  For the {Baum-Welch Reaction System} $(BW(\mathcal{M},h^*,v^*,L),k)$ with permissible rates $k$, if the concentrations of $\alpha,\beta,\gamma$ and $\xi$ species are held fixed at a {\em positive} point then the Maximization step reaction system on $\theta$ and $\psi$ converges to equilibrium exponentially fast.
\end{theorem}
See \ref{pf:EMexpconv2} for proof.
\
\\The following examples demonstrate the behavior of the Chemical Baum-Welch Algorithm.
\begin{example}\label{Ex::1}
  Consider an HMM  $(H,V, \theta,\psi,\pi)$ with two hidden states $H=\{H_1,H_2\}$ and two emitted symbols $V=\{V_1,V_2\}$ where the starting probability is $\pi=(0.6,0.4)$, initial transition probability is $\theta=\begin{bmatrix}0.6 &0.4\\0.3 &0.7\end{bmatrix}$, and initial emission probability is $\psi=\begin{bmatrix}0.5 & 0.5\\0.5 & 0.5 \end{bmatrix}$. Suppose we wish to learn $(\theta,\psi)$ for the following observed sequence: \\ $(V_1,V_1,V_1,V_2,V_1,V_1,V_2,V_2,V_2,V_1,V_2,V_2,V_1,V_1,V_1,V_1,V_2,V_2,V_2,V_1,V_2,V_1,V_1,V_2,V_2)$.
  We initialize the corresponding reaction system by setting species $E_{l,v_l}=1$ and $E_{l,w}=0$ for $w\neq v_l$, and run the dynamics according to deterministic mass-action kinetics. Initial conditions of all species that are not mentioned are chosen to be nonzero, but otherwise at random.
  \begin{figure}[h]
    \centering
    \subfloat[Reaction Network Dynamics]{{\includegraphics[width=6.5cm]{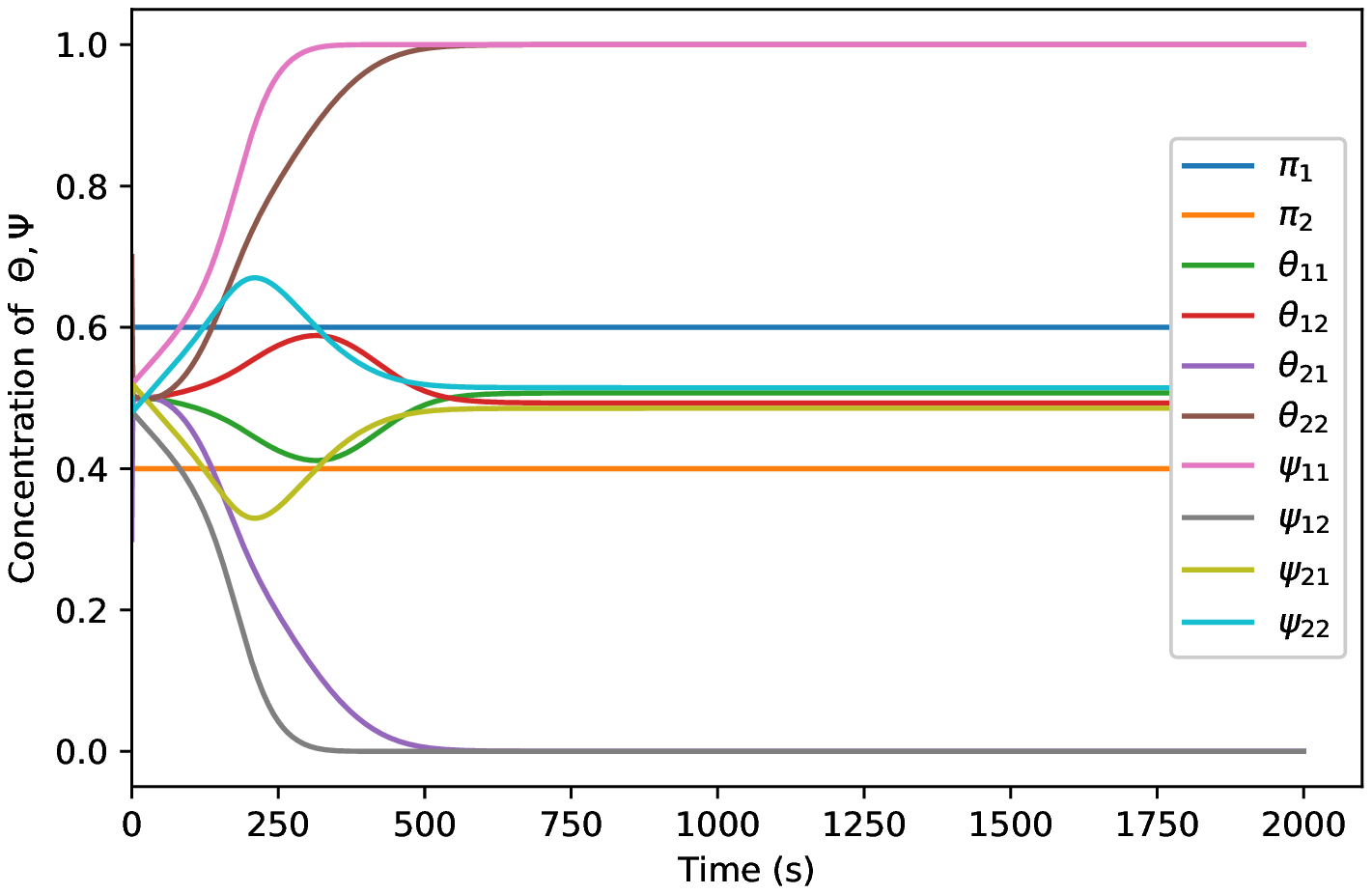} }}
    \subfloat[Baum-Welch Algorithm]{{\includegraphics[width=6.5cm]{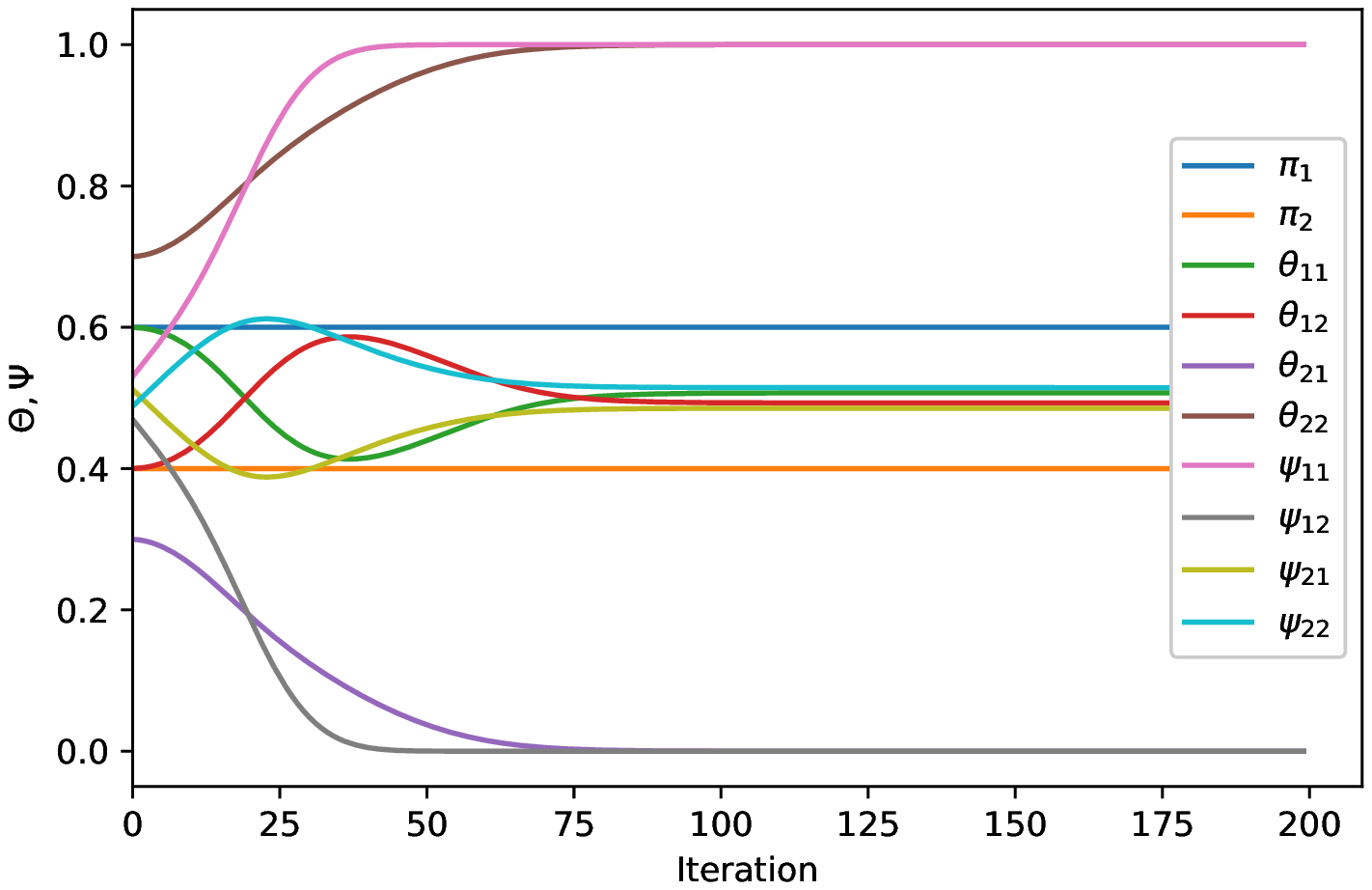} }}
    \caption{Simulation of Example~\ref{Ex::1}: Both simulations are started at exactly the same initial vector. This may not be apparent in the figure because concentration of some species change rapidly at start in the reaction network.}
    \label{fig:Sim1}%
  \end{figure}\\ 
  For our numerical solution, we observe that the reaction network equilibrium point coincides with the Baum-Welch steady state $\hat\theta=\begin{bmatrix}0.5071 &0.4928\\0.0000 &1.0000\end{bmatrix}$ and $\hat\psi=\begin{bmatrix}1.0000 & 0.0000\\  0.4854 & 0.5145 \end{bmatrix}$ (See Figure~\ref{fig:Sim1}).
\end{example}

The next example shows that the Chemical Baum-Welch algorithm can sometimes get stuck at points that are not equilibria for the Baum-Welch algorithm. This is a problem especially for very short sequences, and is probably happening because in such settings, the best model sets many parameters to zero. When many species concentrations are set to $0$, many reactions get turned off, and the network gets stuck away from the desired equilibrium. We believe this will not happen if the HMM generating the observed sequence has nonzero parameters, and the observed sequence is long enough.

\begin{example}	\label{Ex::2} 
  Consider an HMM  $(H,V, \theta,\psi,\pi)$ with two hidden states $H=\{H_1,H_2\}$ and two emitted symbols $V=\{V_1,V_2\}$, initial distribution $\pi=(0.6,0.4)$ is fixed, initial transition probability $\theta=\begin{bmatrix}0.6 &0.4\\0.3 &0.7\end{bmatrix}$, and initial emission probability $\psi=\begin{bmatrix}0.5 & 0.5\\0.5 & 0.5 \end{bmatrix}$. Suppose we wish to learn $(\theta,\psi)$ for the sequence $(V_1,V_2,V_1,V_2,V_1)$.
  We again simulate the corresponding reaction network and also perform the Baum-Welch algorithm for comparison.
  \begin{figure}[h]
    \centering
    \subfloat[Reaction Network Dynamics]{{\includegraphics[width=6.5cm]{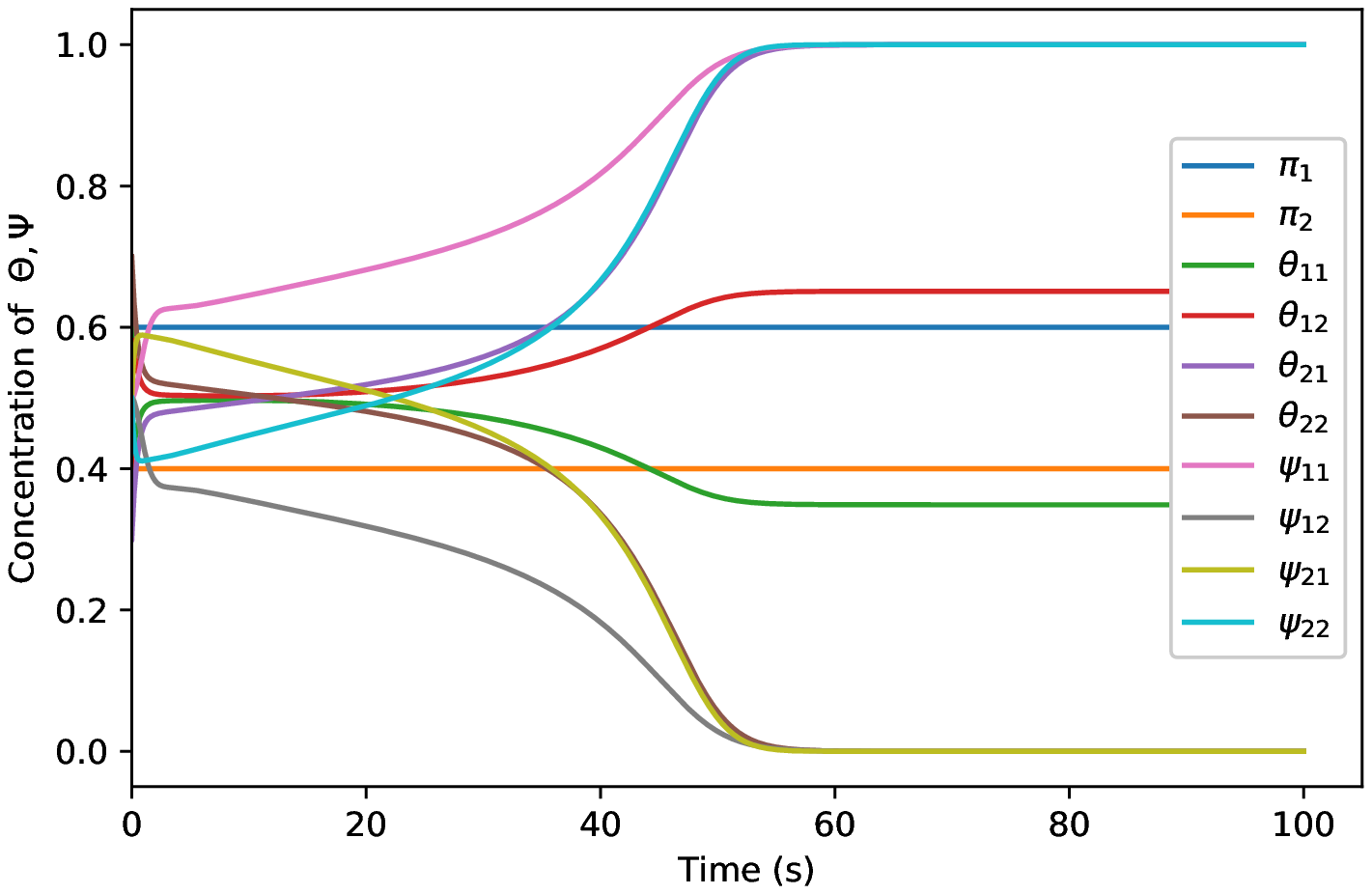} }}
    \subfloat[Baum-Welch Algorithm]{{\includegraphics[width=6.5cm]{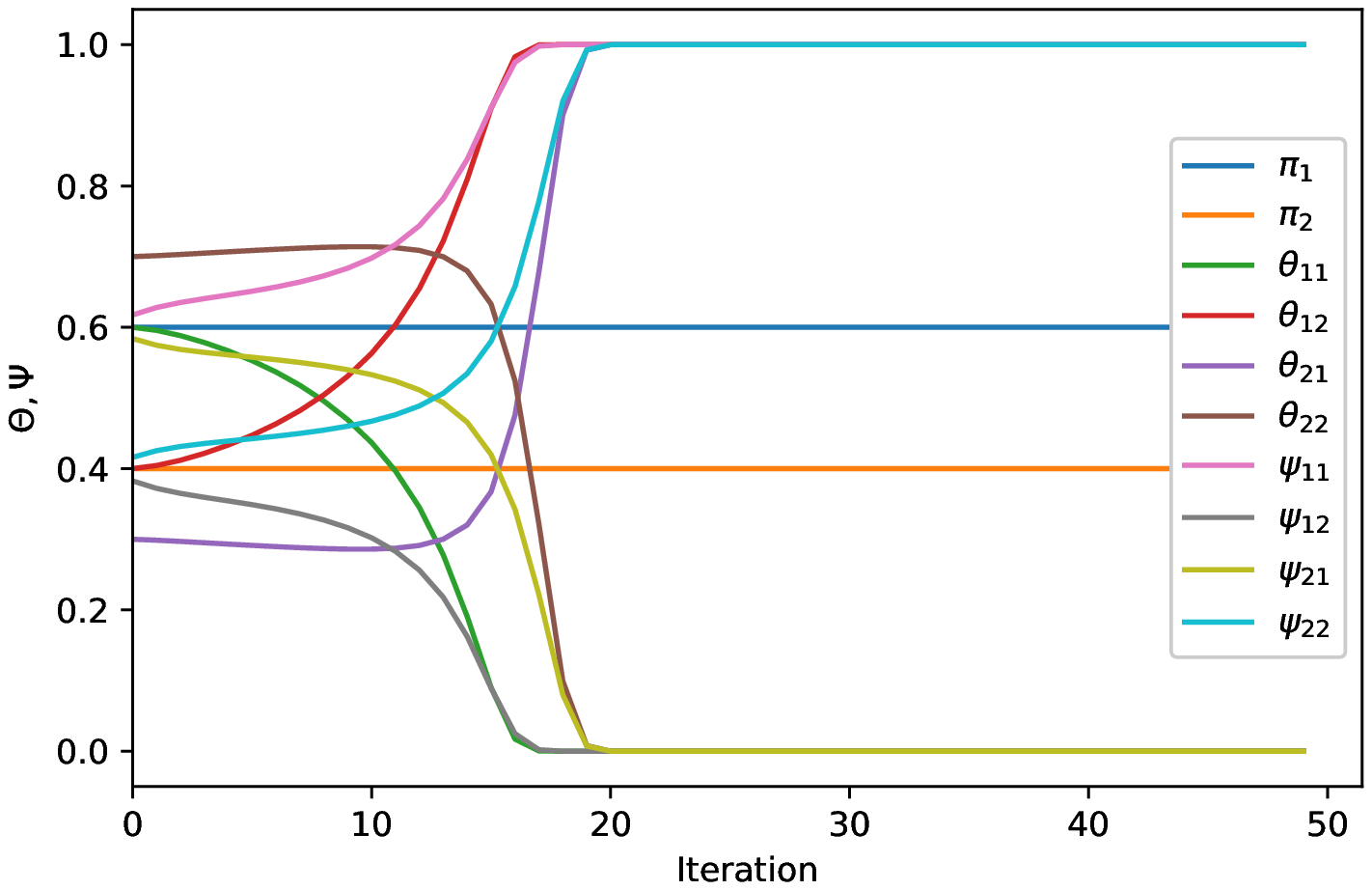} }}
    \caption{Simulation of Example~\ref{Ex::2}: Both simulations are started at exactly the same initial vector. This may not be apparent in the figure because concentration of some species change rapidly at start in the reaction network.}
    \label{fig:Sim2}%
  \end{figure}\\ 
  Figure \ref{fig:Sim2} shows that the reaction network equilibrium point does not coincide with the Baum-Welch steady state. Both converge to $\hat\psi=\begin{bmatrix}1.0000 & 0.0000\\  0.0000 & 1.0000 \end{bmatrix}$. However, the reaction network converges to $\hat\theta=\begin{bmatrix}0.3489 &0.6510\\1.0000 &0.0000\end{bmatrix}$ whereas the Baum-Welch algorithm converges to $\hat\theta=\begin{bmatrix}0.0000 &1.0000\\1.0000 &0.0000\end{bmatrix}$ which happens to be the true maximum likelihood point. Note that this does not contradict Theorem 2 because the fixed point of this system is a boundary point and not a positive fixed point.
\end{example}

\section{Related work}\label{sec::related}
In previous work, our group has shown that reaction networks can perform the Expectation Maximization (EM) algorithm for partially observed log linear statistical models~\cite{virinchi2018reaction,EMposterAbhinav}. That algorithm also applies ``out of the box'' to learning HMM parameters. The problem with that algorithm is that the size of the reaction network would become exponentially large in the length of the sequence, so that even examples like Example~\ref{Ex::1} with an observation sequence of length $25$ would become impractical. In contrast, the scheme we have presented in this paper requires only a linear growth with sequence length. We have obtained the savings by exploiting the graphical structure of HMMs. This allows us to compute the likelihoods $\alpha$ in a ``dynamic programming'' manner, instead of having to explicitly represent each path as a separate species.

Napp and Adams~\cite{napp2013message} have shown how to compute marginals on graphical models with reaction networks. They exploit graphical structure by mimicking belief propagation. Hidden Markov Models can be viewed as a special type of graphical model where there are $2L$ random variables $X_1,X_2,\dots, X_L,Y_1,Y_2,\dots,Y_L$ with the $X$ random variables taking values in $H$ and the $Y$ random variables in $V$. The $X$ random variables form a Markov chain $X_1 \,\rule[.8ex]{10pt}{.5pt}\, X_2 \,\rule[.8ex]{10pt}{.5pt} \dots \rule[.8ex]{10pt}{.5pt}\, X_L$. In addition, there are $L$ edges from $X_l$ to $Y_l$ for $l=1$ to $L$ denoting observations. Specialized to HMMs, the scheme of Napp and Adams would compute the equivalent of steady state values of the $\gamma$ species, performing a version of the E step. They are able to show that true marginals are fixed points of their scheme, which is similar to our Theorem~\ref{thm:bw}. Thus their work may be viewed as the first example of a reaction network scheme that exploits graphical structure to compute E projections. Our E step goes further by proving correctness as well as exponential convergence. Their work also raises the challenge of extending our scheme to all graphical models.

Poole et al.~\cite{poole2017chemical} have described Chemical Boltzmann Machines, which are reaction network schemes whose dynamics reconstructs inference in Boltzmann Machines. This inference can be viewed as a version of E projection. No scheme for learning is presented. The exact schemes presented there are exponentially large. The more realistically sized schemes are presented there without proof. In comparison, our schemes are polynomially sized, provably correct if the equilibrium is positive, and perform both inference and learning for HMMs.

Zechner et al.~\cite{zechner2016molecular} have shown that Kalman filters can be implemented with reactions networks. Kalman filters can be thought of as a version of Hidden Markov Models with continuous hidden states \cite{KalmanRoweis}. It would be instructive to compare their scheme with ours, and note similarities and differences. In passing from position $l$ to position $l+1$ along the sequence, our scheme repeats the same reaction network that updates $\alpha_{l+1}$ using $\alpha_l$ values. It is worth examining if this can be done ``in place'' so that the same species can be reused, and a reaction network can be described that is not tied to the length $L$ of the sequence to be observed. 

Recently Cherry et al.~\cite{cherry2018scaling} have given a brilliant experimental demonstration of learning with DNA molecules. They have empirically demonstrated a DNA molecular system that can classify $9$ types of handwritten digits from the MNIST database. Their approach is based on the notion of ``winner-takes-all'' circuits due to Maass~\cite{maass2000computational} which was originally a proposal for how neural networks in the brain work. Winner-take-all might also be capable of approximating HMM learning, at least in theory~\cite{kappel2014stdp}, and it is worth understanding precisely how such schemes relate to the kind of scheme we have described here. It is conceivable that our scheme could be converted to winner-take-all by getting different species in the same flower to give negative feedback to each other. This might well lead to sampling the most likely path, performing a decoding task similar to the Viterbi algorithm.

\section{Discussion}\label{sec::discuss}
We have described a one-pot one-shot reaction network implementation of the Baum-Welch algorithm. Firstly, this involves proposing a reaction system whose positive fixed points correspond to equilibria of the Baum-Welch Algorithm. Secondly, this involves establishing the conditions under which convergence of solutions to the equilibria can be accomplished. Further, from a practical perspective, it is essential to obtain an implementation that does not rely on repeated iteration of a reaction scheme but only requires to be run once.

As we observe in Remark~\ref{rmk:topology}, there is a whole class of reaction networks that implements the Baum-Welch algorithm. We have proposed one such network and are aware that there are other networks, potentially with more efficient dynamical and analytical properties than the one proposed here. Finding and characterizing efficient reaction networks that can do complicated statistical tasks will likely be of future concern.

We have only discussed deterministic dynamical properties of the network. However, in realistic biological contexts one might imagine that the network is implemented by relatively few molecules such that stochastic effects are significant. Consequently, the study of the Baum-Welch reaction system under stochastic mass-action kinetics is likely to be of interest.

Lastly, we have mentioned the Viterbi algorithm, but have made no attempt to describe how the maximum likelihood sequence can be recovered from our reaction network. This decoding step is likely to be of as much interest for molecular communication systems and molecular multi-agent systems as it is  in more traditional domains of communications and multi-agent reinforcement learning. Because of the inherent stochasticity of reaction networks, there might even be opportunities for list decoding by sampling different paths through the hidden states that have high probability conditioned on the observations. This might give an artificial cell the ability to ``imagine'' different possible realities, and act assuming one of them to be the case, leading to an intrinsic stochasticity and unpredictability to the behavior.

\bibliographystyle{unsrt}
\bibliography{ref.bib}

\begin{thebibliography}{10}

\bibitem{soloveichik2010dna}
David Soloveichik, Georg Seelig, and Erik Winfree.
\newblock {DNA} as a universal substrate for chemical kinetics.
\newblock {\em Proceedings of the National Academy of Sciences},
  107(12):5393--5398, 2010.

\bibitem{srinivas2015programming}
Niranjan Srinivas.
\newblock {\em Programming chemical kinetics: engineering dynamic reaction
  networks with {DNA} strand displacement}.
\newblock PhD thesis, California Institute of Technology, 2015.

\bibitem{qian2011efficient}
Lulu Qian, David Soloveichik, and Erik Winfree.
\newblock {Efficient Turing-universal computation with DNA polymers}.
\newblock In {\em {DNA} computing and molecular programming}, pages 123--140.
  2011.

\bibitem{Cardelli_2011StrandAlgebra}
Luca Cardelli.
\newblock {Strand Algebras for {DNA} Computing}.
\newblock {\em Natural Computing}, 10:407--428, 2011.

\bibitem{lakin2011abstractions}
Matthew~R. Lakin, Simon Youssef, Luca Cardelli, and Andrew Phillips.
\newblock Abstractions for {DNA} circuit design.
\newblock {\em Journal of The Royal Society Interface}, 9(68):470--486, 2011.

\bibitem{cardelli2013two}
Luca Cardelli.
\newblock Two-domain {DNA} strand displacement.
\newblock {\em Mathematical Structures in Computer Science}, 23:02, 2013.

\bibitem{chen2013programmable}
Yuan-Jyue Chen, Neil Dalchau, Niranjan Srinivas, Andrew Phillips, Luca
  Cardelli, David Soloveichik, and Georg Seelig.
\newblock Programmable chemical controllers made from {DNA}.
\newblock {\em Nature nanotechnology}, 8(10):755--762, 2013.

\bibitem{LAKIN201621}
Matthew~R. Lakin, Darko Stefanovic, and Andrew Phillips.
\newblock Modular verification of chemical reaction network encodings via
  serializability analysis.
\newblock {\em Theoretical Computer Science}, 632:21--42, 2016.

\bibitem{Srinivasetal2052}
Niranjan Srinivas, James Parkin, Georg Seelig, Erik Winfree, and David
  Soloveichik.
\newblock Enzyme-free nucleic acid dynamical systems.
\newblock {\em Science}, 358:6369, 2017.

\bibitem{cherry2018scaling}
Kevin~M. Cherry and Lulu Qian.
\newblock Scaling up molecular pattern recognition with {DNA}-based
  winner-take-all neural networks.
\newblock {\em Nature}, 559:7714, 2018.

\bibitem{zechner2016molecular}
Christoph Zechner, Georg Seelig, Marc Rullan, and Mustafa Khammash.
\newblock Molecular circuits for dynamic noise filtering.
\newblock {\em Proceedings of the National Academy of Sciences},
  113(17):4729--4734, 2016.

\bibitem{badelt2017general}
Stefan Badelt, Seung~Woo Shin, Robert~F. Johnson, Qing Dong, Chris Thachuk, and
  Erik Winfree.
\newblock A general-purpose {CRN}-to-{DSD} compiler with formal verification,
  optimization, and simulation capabilities.
\newblock In {\em International Conference on {DNA}-Based Computers}, pages
  232--248, 2017.

\bibitem{lakin2011visual}
Matthew~R. Lakin, Simon Youssef, Filippo Polo, Stephen Emmott, and Andrew
  Phillips.
\newblock Visual {DSD}: a design and analysis tool for dna strand displacement
  systems.
\newblock {\em Bioinformatics}, 27(22):3211--3213, 2011.

\bibitem{hjelmfelt1991chemical}
Allen Hjelmfelt, Edward~D. Weinberger, and John Ross.
\newblock {Chemical implementation of neural networks and Turing machines}.
\newblock {\em Proceedings of the National Academy of Sciences},
  88(24):10983--10987, 1991.

\bibitem{buisman2009computing}
H~J Buisman, Huub~MM ten Eikelder, Peter~AJ Hilbers, and Anthony~ML Liekens.
\newblock Computing algebraic functions with biochemical reaction networks.
\newblock {\em Artificial life}, 15(1):5--19, 2009.

\bibitem{Klavins_2011Biomolecular}
Kevin Oishi and Eric Klavins.
\newblock {Biomolecular implementation of linear I/O systems}.
\newblock {\em Systems Biology, IET}, 5(4):252--260, 2011.

\bibitem{soloveichik2008computation}
David Soloveichik, Matthew Cook, Erik Winfree, and Jehoshua Bruck.
\newblock Computation with finite stochastic chemical reaction networks.
\newblock {\em natural computing}, 7(4):615--633, 2008.

\bibitem{chen2014deterministic}
Ho-Lin Chen, David Doty, and David Soloveichik.
\newblock Deterministic function computation with chemical reaction networks.
\newblock {\em Natural computing}, 13(4):517--534, 2014.

\bibitem{qian2011scaling}
Lulu Qian and Erik Winfree.
\newblock Scaling up digital circuit computation with {DNA} strand displacement
  cascades.
\newblock {\em Science}, 332(6034):1196--1201, 2011.

\bibitem{napp2013message}
Nils~E. Napp and Ryan~P. Adams.
\newblock Message passing inference with chemical reaction networks.
\newblock {\em Advances in Neural Information Processing Systems}, pages
  2247--2255, 2013.

\bibitem{Winfree_2011Neural}
Lulu Qian, Erik Winfree, and Jehoshua Bruck.
\newblock {Neural Network Computation with DNA Strand Displacement Cascades}.
\newblock {\em Nature}, 475(7356):368--372, 2011.

\bibitem{cardelli2018chemical}
Luca Cardelli, Marta Kwiatkowska, and Max Whitby.
\newblock Chemical reaction network designs for asynchronous logic circuits.
\newblock {\em Natural computing}, 17(1):109--130, 2018.

\bibitem{gopalkrishnan2016scheme}
Manoj Gopalkrishnan.
\newblock A scheme for molecular computation of maximum likelihood estimators
  for log-linear models.
\newblock In {\em 22nd International Conference on {DNA} Computing and
  Molecular Programming}, pages 3--18, 2016.

\bibitem{virinchi2017stochastic}
Muppirala~Viswa Virinchi, Abhishek Behera, and Manoj Gopalkrishnan.
\newblock A stochastic molecular scheme for an artificial cell to infer its
  environment from partial observations.
\newblock In {\em 23rd International Conference on {DNA} Computing and
  Molecular Programming}, 2017.

\bibitem{virinchi2018reaction}
Muppirala~Viswa Virinchi, Abhishek Behera, and Manoj Gopalkrishnan.
\newblock A reaction network scheme which implements the {EM} algorithm.
\newblock In {\em International Conference on DNA Computing and Molecular
  Programming}, pages 189--207, 2018.

\bibitem{Amari2016}
Shunichi Amari.
\newblock {\em Information geometry and its applications}.
\newblock Springer, 2016.

\bibitem{csiszar2003information}
Imre Csisz{\'a}r and Frantisek Matus.
\newblock Information projections revisited.
\newblock {\em IEEE Transactions on Information Theory}, 49(6):1474--1490,
  2003.

\bibitem{jaynes1957information}
Edwin~T. Jaynes.
\newblock Information theory and statistical mechanics.
\newblock {\em Physical review}, 106:4, 1957.

\bibitem{shin2004synthetic}
Jong-Shik Shin and Niles~A. Pierce.
\newblock A synthetic {DNA} walker for molecular transport.
\newblock {\em Journal of the American Chemical Society}, 126(35):10834--10835,
  2004.

\bibitem{Reif_2003WalkerRoller}
John Reif.
\newblock {The Design of Autonomous DNA Nano-mechanical Devices: Walking and
  Rolling DNA}.
\newblock {\em DNA Computing}, pages 439--461, 2003.

\bibitem{Seeman_2004Walker}
William Sherman and Nadrian Seeman.
\newblock {A Precisely Controlled DNA Biped Walking Device}.
\newblock {\em Nano Letters}, 4:1203--1207, 2004.

\bibitem{cover2012elements}
Thomas~M. Cover and Joy~A. Thomas.
\newblock {\em Elements of information theory}.
\newblock Wiley, John \& Sons, 2012.

\bibitem{18626}
L.~R. {Rabiner}.
\newblock {A tutorial on hidden Markov models and selected applications in
  speech recognition}.
\newblock {\em Proceedings of the IEEE}, 77(2):257--286, February 1989.

\bibitem{juang1991hidden}
Biing~Hwang Juang and Laurence~R. Rabiner.
\newblock {Hidden Markov models for speech recognition}.
\newblock {\em Technometrics}, 33(3):251--272, 1991.

\bibitem{feinberg72chemical}
Martin Feinberg.
\newblock On chemical kinetics of a certain class.
\newblock {\em Arch. Rational Mech. Anal}, 46, 1972.

\bibitem{horn72necessary}
Friedrich J.~M. Horn.
\newblock Necessary and sufficient conditions for complex balancing in chemical
  kinetics.
\newblock {\em Arch. Rational Mech. Anal}, 49, 1972.

\bibitem{Fein79}
Martin Feinberg.
\newblock Lectures on chemical reaction networks.
  \url{http://www.che.eng.ohio-state.edu/FEINBERG/LecturesOnReactionNetworks/}.
\newblock 1979.

\bibitem{Manoj_2011Catalysis}
Manoj Gopalkrishnan.
\newblock {Catalysis in Reaction Networks}.
\newblock {\em Bulletin of Mathematical Biology}, 73(12):2962--2982, 2011.

\bibitem{anderson2010product}
David~F. Anderson, Gheorghe Craciun, and Thomas~G. Kurtz.
\newblock Product-form stationary distributions for deficiency zero chemical
  reaction networks.
\newblock {\em Bulletin of mathematical biology}, 72(8):1947--1970, 2010.

\bibitem{tu2006metabolic}
Benjamin~P. Tu and Steven~L. McKnight.
\newblock Metabolic cycles as an underlying basis of biological oscillations.
\newblock {\em Nature reviews Molecular cell biology}, 7:9, 2006.

\bibitem{mclachlan2007algorithm}
Geoffrey McLachlan and Thriyambakam Krishnan.
\newblock {\em The {EM} algorithm and extensions}, volume 382.
\newblock Wiley, John \& Sons, 2007.

\bibitem{EMposterAbhinav}
Abhinav Singh and Manoj Gopalkrishnan.
\newblock {EM Algorithm with DNA Molecules}.
\newblock In {\em Poster Presentations of the 24-th edition of International
  Conference on DNA Computing and Molecular Programming}, 2018.

\bibitem{poole2017chemical}
William Poole, Andr{\'e}s Ortiz-Munoz, Abhishek Behera, Nick~S. Jones,
  Thomas~E. Ouldridge, Erik Winfree, and Manoj Gopalkrishnan.
\newblock Chemical boltzmann machines.
\newblock In {\em International Conference on DNA-Based Computers}, pages
  210--231, 2017.

\bibitem{KalmanRoweis}
S.~Roweis and Z.~Ghahramani.
\newblock {A unifying review of linear Gaussian models}.
\newblock {\em Neural Computation}, 11(2):305--345, 1999.

\bibitem{maass2000computational}
Wolfgang Maass.
\newblock On the computational power of winner-take-all.
\newblock {\em Neural computation}, 12(11):2519--2535, 2000.

\bibitem{kappel2014stdp}
David Kappel, Bernhard Nessler, and Wolfgang Maass.
\newblock {STDP installs in winner-take-all circuits an online approximation to
  hidden Markov model learning}.
\newblock {\em PLoS computational biology}, 10:3, 2014.

\end{thebibliography}

\begin{appendices}
  \renewcommand{\thesection}{\appendixname~\Alph{section}}
  \section{}

  \subsection{Comparing points of equilibria}
  We will now prove Theorem~\ref{thm:bw}. For the sake of convenience, we first recall the statement.\\\
  \par\noindent\textbf{Theorem~\ref{thm:bw}.}
  {\em Every fixed point of the Baum-Welch algorithm for an HMM $\mathcal{M}=(H,V,\theta,\psi,\pi)$ is a fixed point for the corresponding Chemical Baum-Welch Algorithm with permissible rates $k$.}
  \begin{proof}\label{pf:bw}
    Consider a point $\Phi = (\alpha',\beta',\gamma',\xi',\theta',\psi')$ with $\alpha',\beta',\gamma'\in\mathbb{R}_{\ge0}^{L\times H}$, $\xi'\in\mathbb{R}_{\ge0}^{(L-1)\times H\times H}$, $\theta'\in\mathbb{R}_{\ge0}^{H\times H}$ and $\psi'\in\mathbb{R}_{\ge0}^{H\times V}$. If $\Phi$ is a fixed point of Baum-Welch Algorithm then it must satisfy:
    \begin{itemize}
    \item $\alpha'_{1h} = \pi_h\psi'_{h v_1}$ for all $h\in H$. Then for the chemical Baum-Welch Algorithm we have \[\dot\alpha_{1h}\big|_\Phi=k^\alpha_{1h}\left({\alpha'_{1h^*} }{ \pi_h\psi'_{h v_1}}-{\alpha'_{1h} }{ \pi_{h^*}\psi'_{h^* v_1}}\right)=0\] for all $h\in H\setminus\{h^*\}$ and $\dot\alpha_{1h^*}\big|_\Phi=-\sum_{h\not=h^*}\dot\alpha_{1h}\big|_\Phi=0$.
    \item $\alpha'_{l h}=\sum_{g\in H} \alpha'_{l-1,g}\theta'_{gh}\psi'_{h v_l}${ for all }$h\in H${ and }$l=2,\ldots,L$ . Then for the chemical Baum-Welch Algorithm we have
      \[\dot\alpha_{lh}\big|_\Phi=k^\alpha_{lh}\left({\alpha'_{l h^*}}{\sum_{g\in H} \alpha'_{l-1,g}\theta'_{gh}\psi'_{h v_l}}-{\alpha'_{l h}}{\sum_{g\in H} \alpha'_{l-1,g}\theta'_{gh^*}\psi'_{h^* v_l}}\right)=0\]
      { for all }$h\in H\setminus\{h^*\}${ and }$l=2,\ldots,L$ and $\dot\alpha_{lh^*}\big|_\Phi=-\sum_{h\not=h^*}\dot\alpha_{lh}\big|_\Phi=0$.
    \item $\beta'_{l h}=\sum_{g\in H} \theta'_{hg}\psi'_{g v_{l+1}}\beta'_{l+1,g}${ for all }$h\in H${ and }$l=1,\ldots,L-1$. Then for the chemical Baum-Welch Algorithm we have
      \[\dot\beta_{lh}\big|_\Phi=k^\beta_{lh}\left({\beta'_{l h^*}}{\sum_{g\in H} \theta'_{h^*g}\psi'_{g v_{l+1}}\beta'_{l+1,g}}-{\beta'_{l h}}{\sum_{g\in H} \theta'_{hg}\psi'_{g v_{l+1}}\beta'_{l+1,g}}\right)=0\]
      { for all }$h\in H\setminus\{h^*\}${ and }$l=1,2,\dots,L-1$ and $\dot\beta_{lh^*}\big|_\Phi=-\sum_{h\not=h^*}\dot\beta_{lh}\big|_\Phi=0$.
    \item $\gamma'_l(h)=\frac{\alpha'_{l h}\beta'_{l h}}{\sum_{g\in H} \alpha'_{l g}\beta'_{l g}}${ for all }$h\in H${ and }$l=1,2,\dots,L-1$.  Then for the chemical Baum-Welch Algorithm we have
      \[\dot\gamma_{lh}=k^\gamma_{lh}\left({\gamma'_{lh^*}}{\alpha'_{l h}\beta'_{l h}}-{\gamma'_{lh}}{\alpha'_{l h^*}\beta'_{l h^*}}\right)=0\]
      { for all }$h\in H\setminus\{h^*\}${ and }$l=1,2,\dots,L-1$ and $\dot\gamma_{lh^*}=-\sum_{h\not=h^*}\dot\gamma_{lh}=0$.
    \item $\xi'_l(g,h)=\frac{\alpha'_{l g}\theta'_{gh}\psi'_{ hv_{l+1}}\beta'_{l+1,h}}{\sum_{f\in H} \alpha'_{l f}\beta'_{l f}}\text{ for all } g,h\in H\text{ and }l=1,\ldots,L-1 $.  Then for the chemical Baum-Welch Algorithm we have
      \[ \dot\xi_{lgh}\big|_\Phi=k^\xi_{lgh}\left({\xi'_{lh^*h^*}}{\alpha'_{l g}\theta'_{gh}\psi'_{ hv_{l+1}}\beta'_{l+1,h}}-{\xi'_{lgh}}{\alpha'_{l h^*}\theta'_{h^*h^*}\psi'_{ h^*v_{l+1}}\beta'_{l+1,h^*}}\right)=0\]
      $\text{ for all } g,h\in H\times H\setminus\{(h^*,h^*)\}\text{ and } l=1,\ldots,L-1$ and $\dot\xi_{lh^*h^*}\big|_\Phi=-\sum_{(g,h)\not=(h^*,h^*)}\dot\xi_{lgh}\big|_\Phi=0$.
    \item $\theta'_{gh}=\frac{\sum_{l=1}^{L-1}\xi'_l(g,h)}{\sum_{l=1}^{L-1}\sum_{f\in H}\xi'_l(g,f)}\text{ for all }g,h\in H $.  Then for the chemical Baum-Welch Algorithm we have
      \[ \dot\theta_{gh}\big|_\Phi=k^\theta_{gh}\left({\theta'_{gh^*}}{\sum_{l=1}^{L-1}\xi'_{lgh}}-{\theta'_{gh}}{\sum_{l=1}^{L-1}\xi'_l(g,h^*)}\right)=0\]
      $\text{ for all }g\in H\text{ and }h\in H\setminus\{h^*\}$ and $\dot\theta_{gh^*}\big|_\Phi=-\sum_{h\not=h^*}\dot\theta_{gh}\big|_\Phi=0 $.
    \item $\psi'_{hw}=\frac{\sum_{l=1}^L\gamma'_l(h)\delta_{w,v_l}}{\sum_{l=1}^L\gamma'_l(h)}\text{ for all }h\in H\text{ and }w\in V $.  Then for the chemical Baum-Welch Algorithm we have
      \[\dot\theta_{gh}\big|_\Phi=k^\theta_{gh}\left({\theta'_{gh^*}}{\sum_{l=1}^{L-1}\xi'_{lgh}}-{\theta'_{gh}}{\sum_{l=1}^{L-1}\xi'_l(g,h^*)}\right) =0\]
      $\text{ for all }h\in H\text{ and }w\in V\setminus\{v^*\}$ and $\dot\psi_{hw^*}\big|_\Phi=-\sum_{w\not=w^*}\dot\psi_{hw}\big|_\Phi =0$.
    \end{itemize}
    So $\Phi$ is fixed point of the chemical Baum-Welch Algorithm.
  \qed\end{proof}
  
  We will now prove Theorem~\ref{thm:bwcrn}. For the sake of convenience, we first recall the statement.\\\
  \par\noindent\textbf{Theorem~\ref{thm:bwcrn}.}
  {\em Every positive fixed point for the Chemical Baum-Welch Algorithm on a {Baum Welch Reaction system} $(BW(\mathcal{M},h^*,v^*,L),k)$ with permissible rate $k$ is a fixed point for the Baum-Welch algorithm for the HMM $\mathcal{M}=(H,V,\theta,\psi,\pi)$.}
  
  \begin{proof}\label{pf:bwcrn}
    Consider a \emph{positive} point $\Phi=(\alpha',\beta',\gamma',\xi',\theta',\psi')$ with $\alpha',\beta',\gamma'\in\mathbb{R}_{>0}^{L\times H}$, $\xi'\in\mathbb{R}_{>0}^{(L-1)\times H\times H}$, $\theta'\in\mathbb{R}_{>0}^{H\times H}$ and $\psi'\in\mathbb{R}_{>0}^{H\times V}$. If $\Phi$ is a fixed point for the Chemical Baum-Welch Algorithm then we must have:
    \begin{itemize}
    \item $\dot\alpha_{1h}\big|_\Phi=0$ for all $h\in H$. This implies ${\alpha'_{1h} }\times{ \pi_{h^*}\psi'_{h^* v_1}}={\alpha'_{1h^*} }\times{ \pi_h\psi'_{h v_1}}\quad\text{ for all }h\in H\setminus{h^*}$. Since $\Phi$ is positive, this implies \[{\alpha'_{1h} }=\left({ \pi_h\psi'_{h v_1}}\right)\frac{\sum_{f\in H}\alpha'_{1f} }{\sum_{f\in H} \pi_{f}\psi'_{f v_1}}\quad\text{ for all }h\in H\]
    \item $\dot\alpha_{lh}\big|_\Phi=0$ for all $h\in H$ and $l=2,\ldots,L$. This implies ${\alpha'_{l h}}\times{\sum_{g\in H} \alpha'_{l-1,g}\theta'_{gh^*}\psi'_{h^* v_l}}={\alpha'_{l h^*}}\times{\sum_{g\in H} \alpha'_{l-1,g}\theta'_{gh}\psi'_{h v_l}}\quad\text{ for all }h\in H\setminus\{h^*\}\text{ and }l=2,\ldots,L$. Since $\Phi$ is positive, this implies \[{\alpha'_{l h}}=\left({\sum_{g\in H} \alpha'_{l-1,g}\theta'_{gh}\psi'_{h v_l}}\right)\frac{\sum_{f\in H}\alpha'_{l f}}{\sum_{f,g\in H} \alpha'_{l-1,g}\theta'_{gf}\psi'_{f v_l}}\quad\text{ for all }h\in H\text{ and }l=2,\ldots,L\]
    \item $\dot\beta_{lh}\big|_\Phi=0$ for all $h\in H$ and $l=1,\ldots,L$. This implies ${\beta'_{l h}}\times{\sum_{g\in H} \theta'_{hg}\psi'_{g v_{l+1}}\beta'_{l+1,g}}={\beta'_{l h^*}}\times{\sum_{g\in H} \theta'_{h^*g}\psi'_{g v_{l+1}}\beta'_{l+1,g}}\quad\text{ for all }h\in H\setminus\{h^*\}\text{ and }l=1,\ldots,L-1$. Since $\Phi$ is positive, this implies \[{\beta'_{l h}}=\left({\sum_{g\in H} \theta'_{hg}\psi'_{g v_{l+1}}\beta'_{l+1,g}}\right)\frac{\sum_{f\in H}\beta'_{l f}}{\sum_{f,g\in H} \theta'_{fg}\psi'_{g v_{l+1}}\beta'_{l+1,g}}\quad\text{ for all }h\in H\text{ and }l=1,\ldots,L-1\]
    \item $\dot\gamma_{lh}\big|_\Phi=0$ for all $h\in H$ and $l=1,\ldots,L$. This implies ${\gamma'_l(h)}\times{\alpha'_{l h^*}\beta'_{l h^*}}={\gamma'_l(h^*)}\times{\alpha'_{l h}\beta'_{l h}}\quad\text{ for all }h\in H\setminus\{h^*\}\text{ and }l=1,2,\dots,L-1$. Since $\Phi$ is positive, this implies \[{\gamma'_{lh}}=\left(\frac{\alpha'_{l h}\beta'_{l h}}{\sum_{g\in H}\alpha'_{l g}\beta'_{l g}}\right){\sum_{g\in H}\gamma'_{lg}}\quad\text{ for all }h\in H\text{ and }l=1,2,\dots,L-1\]
    \item $\dot\xi_{lgh}\big|_\Phi=0$ for all $g,h\in H$ and $l=1,\ldots,L-1$. This implies ${\xi'_l(g,h)}\times{\alpha'_{l h^*}\theta'_{h^*h^*}\psi'_{ h^*v_{l+1}}\beta'_{l+1,h^*}}={\xi'_l(h^*,h^*)}\times{\alpha'_{l g}\theta'_{gh}\psi'_{ hv_{l+1}}\beta'_{l+1,h}}$ $\text{ for all } g,h\in H\times H\setminus\{(h^*,h^*)\}\text{ and } l=1,\ldots,L-1$. Since $\Phi$ is positive, this implies \[{\xi'_{lgh}}=\left(\frac{\alpha'_{l g}\theta'_{gh}\psi'_{ hv_{l+1}}\beta'_{l+1,h}}{\sum_{e,f\in H}\alpha'_{l f}\theta'_{ef}\psi'_{ fv_{l+1}}\beta'_{l+1,f}}\right){\sum_{e,f\in H}\xi'_{lef}} \quad\text{ for all } g,h\in H\times H\text{ and } l=1,\ldots,L-1\]
    \item $\dot\theta_{gh}\big|_\Phi=0$ for all $g,h\in H$. This implies ${\theta'_{gh}}\times{\sum_{l=1}^{L-1}\xi'_l(g,h^*)}= {\theta'_{gh^*}}\times{\sum_{l=1}^{L-1}\xi'_l(g,h)}\quad\text{ for all }g\in H\text{ and }h\in H\setminus\{h^*\}$. Since $\Phi$ is positive, this implies \[{\theta'_{gh}}=\left(\frac{\sum_{l=1}^{L-1}\xi'_{lgh}}{\sum_{f\in H}\sum_{l=1}^{L-1}\xi'_{lgf}}\right){\sum_{f\in H}\theta'_{gf}}\quad\text{ for all }g\in H\text{ and }h\in H\]
    \item $\dot\psi_{hw}\big|_\Phi=0$ for all $h\in H$ and $w\in V$. This implies ${\psi'_{hw}}\times{\sum_{l=1}^L\gamma'_l(h)\delta_{v^*,v_l}} = {\psi'_{hv^*}}\times{\sum_{l=1}^L\gamma'_l(h)\delta_{w,v_l}}\quad\text{ for all }h\in H\text{ and }w\in V\setminus\{v^*\}$. Since $\Phi$ is positive, $E_{lv}=\delta_{v,v_l}$ and $\sum_{v\in V}\delta_{v,v_l}=1$ this implies \[{\psi'_{hw}}=\left(\frac{\sum_{l=1}^L\gamma'_{lh}\delta_{w,v_l}}{\sum_{l=1}^L\gamma'_{lh}}\right){\sum_{v\in V}\psi'_{hv}}\quad\text{ for all }h\in H\text{ and }w\in V\]
    \end{itemize}
    Because of the relaxation we get by Remark~\ref{rmk:scale}, the point $\Phi$ qualifies as a fixed point of the Baum-Welch algorithm.  
  \qed\end{proof}

  \subsection{Rate of convergence analysis}
  In this section we will prove Theorem~\ref{th:EMexpconv1} and Theorem~\ref{th:EMexpconv2}, but first we will state and prove two useful lemmas.
  
  \begin{lemma}\label{lem:kernel}
    Let $A$ be an arbitrary $n\times n$ matrix. Let $W$ be an $r\times n$ matrix comprising of $r$ linearly independent left kernel vectors of $A$ so that  $WA=0_{r,n}$, where $0_{i,j}$ denotes a $i\times j$ matrix with all entries zero. Further suppose $W$ is in the row reduced form, that it is, \[W=\begin{pmatrix}  W' & I_r\end{pmatrix}\] where $I_j$ denotes the $j\times j$ identity matrix and $W'$ is a $ r\times (n-r)$ matrix. Let $A$ be given as \[A=\begin{pmatrix} A_{11} & A_{12} \\ A_{21} & A_{22} \end{pmatrix},\] where $A_{11}$ is a $(n-r)\times(n-r)$ matrix, $A_{12}$ is a $ (n-r)\times r$ matrix, $A_{21}$ is a $r\times (n-r)$ matrix, and $A_{22}$ is a $r\times r$ matrix.
    
    Then the $n-r$ eigenvalues (with multiplicity) of the matrix $A_{11} -A_{12}W'$ are the same as the eigenvalues of $A$ except for $r$ zero eigenvalues.
  \end{lemma}
  \begin{proof}
    Consider the $n\times n$ invertible matrix $P$ given by 
    $$P=\begin{pmatrix} I_{n-r} &0_{n-r,r} \\ W'  & I_r \end{pmatrix}, \quad P^{-1}=\begin{pmatrix} I_{n-r} &0_{n-r,r} \\ -W'   & I_r\end{pmatrix},$$
    with determinant $\det(P)=\det(P^{-1})=1$. We have
    \begin{align*}
      PAP^{-1}&= \begin{pmatrix}  I_{n-r} &0_{n-r,r} \\ W'  & I_r  \end{pmatrix}\begin{pmatrix} A_{11} & A_{12} \\ A_{21} & A_{22} \end{pmatrix}\begin{pmatrix} I_{n-r} &0_{n-r,r} \\ -W'   & I_r\end{pmatrix}\\
              &=\begin{pmatrix}  A_{11} & A_{12} \\ 0_{r,n-r} & 0_{r,r} \end{pmatrix}\begin{pmatrix} I_{n-r} &0_{n-r,r} \\ -W'   & I_r\end{pmatrix} \\
              &=\begin{pmatrix}   A_{11} -A_{12}W' & A_{12}  \\  0_{r,n-r} & 0_{r,r}  \end{pmatrix}
    \end{align*}
    This implies that the characteristic polynomial of $A$ fulfils
    \begin{align*}
      \det(A-\lambda I_n) &=\det(P)\det(A-\lambda I_n)\det(P^{-1}) = \det(PAP^{-1}-\lambda I_n )\\
                          &=\det\begin{pmatrix}  A_{11} -A_{12}W' - \lambda I_{n-r} & A_{12}  \\  0_{r,n-r} & 0_{r,r}   -\lambda I_r\end{pmatrix} \\
                          &= (-1)^r\lambda^r \det(A_{11} -A_{12}W' -\lambda I_{n-r}),
    \end{align*}
    and the statement follows.
  \qed\end{proof}
  
  Now we revisit the observation that every reaction in the Baum-Welch reaction network is a monomolecular transformation catalyzed by a set of species.  For the purposes of our analysis, each reaction can be abstracted as a monomolecular reaction with time varying rate constants. This prompts us to consider the following monomolecular reaction system with $n$ species $X_1,\ldots,X_n$ and $m$ reactions
  \[X_{r_j}\cee{->[k_j(t)]} X_{p_j},\quad \text{for}\quad j=1,\ldots,m,\]
  where $r_j\not=p_j$, and $r_j,p_j\in\{1,\ldots,n\}$, and $k_j(t)$, $j=1,\ldots,m$, are mass-action reaction rate constants, possibly depending on time. We assume $k_j(t)>0$ for  $t\ge 0$ and let $k(t)=(k_1(t),\ldots,k_{m}(t))$ be the vector of  reaction rate constants.  Furthermore, we assume there is at most one reaction $j$ such that $(r_j,p_j)=(r,p)\in\{1,\ldots,n\}^2$ and that the reaction network is strongly connected. The later means there is a reaction path from any species $X_i$ to any other species $X_{i'}$. (In reaction network terms it means the network is weakly reversible and deficiency zero.)
  
  The mass action kinetics of this reaction system is given by the ODE system
  \begin{align*}
    \dot{x}_i&=-x_i\!\!\sum_{j=1\colon r_j=i}^{m}k_j(t)  +\sum_{j=1\colon p_j=i}^{m}k_j(t) x_{r_j}, \qquad i=1,\ldots,n.
  \end{align*}
  Define the $n\times n$ matrix $A(t)=(a_{ii'}(t))_{i,i'=1,\ldots,n}$ by
  \begin{equation}\label{eq:matrixA}
    \begin{aligned}
      a_{ii}(t)&=-\sum_{j=1\colon r_j=i}^{m}k_j(t), \\
      a_{ii'}(t)&=k_j(t),\quad\text{if there is }j\in\{1,\ldots,m\}\,\,\text{such that}\,\, (r_j,p_j)=(i',i).
    \end{aligned}
  \end{equation}
  Then the ODE system might be written as
  \begin{equation}\label{eq:equiv}
    \dot{x}= A(t) x.
  \end{equation}
  Note that the column sums of $A(t)$ are zero, implying that $\sum_{i=1}^n x_i$ is conserved.

  \begin{lemma}\label{lem:expconv}
    Assume $k(t)$ for $t\ge 0$, converges exponentially fast towards $k=(k_1,\ldots,k_{m})\in\R^{m}_{>0}$  as $t\to\infty$, that is, there exists $\gamma_1>0$ and $K_1>0$ such that
    \[\norm{ k(t)-k}\le  K_1  e^{-\gamma_1 t}\quad \text{for}\quad t\ge 0.\]
    Let $A(t)$ be the matrix as defined in equation~\ref{eq:matrixA}. And let $A$ be the matrix obtained with $k$ inserted for $k(t)$ in the matrix $A(t)$ that is, $A=\lim_{t\to\infty} A(t)$.
    
    Then solutions of ODE system $\dot{x}=A(t)x$ starting at $x(0)\in\mathbb{R}^n_{\ge 0}$ converges exponentially fast towards the equilibrium $a\in\mathbb{R}^n_{>0}$ of the ODE system $\dot{x}=Ax$ starting at $x(0)\in\mathbb{R}^n_{\ge 0}$, that is, there exists $\gamma>0$ and $K>0$ such that \[\norm{ x(t)-a}\le  K  e^{-\gamma t}\quad \text{for}\quad t\ge 0.\]
  \end{lemma}
  
  \begin{proof}
    We will first rephrase the ODE system such that standard theory is applicable. Let rank of $A$ be $n-r$. Let $W$ be as defined in Lemma~\ref{lem:kernel}, that is, $W$ be an $r\times n$ matrix comprising of $r$ linearly independent left kernel vectors of $A$ so that  $WA=0$. Here since rank of $A$ is $n-r$, the rows of $W$ would form a basis for the left kernel of $A$. And as in Lemma~\ref{lem:kernel}, further suppose $W$ is in the row reduced form, that is, \[W=\begin{pmatrix}  W' & I_r\end{pmatrix}.\] Then 
    \begin{equation}
      \dot{x}=A x
    \end{equation}
    is a linear dynamical system with $r$ conservation laws (one for each row of $W$). Let $Wx(0)=T\in\mathbb{R}^r$ be the vector of conserved amounts.  Let $\hat x=(x_1,\ldots,x_{n-r})$ and $\tilde x=(x_{n-r+1},\ldots,x_{n})$.  We will consider the (equivalent) dynamical system in which $r$  variables  are eliminated, expressed through the conservation laws 
    \[T=Wx=\begin{pmatrix} W' & I_r \end{pmatrix}x, \quad \text{or} \quad \tilde x=T-W' \hat x.\]
    As in Lemma~\ref{lem:kernel}, let $A$ be given as \[A=\begin{pmatrix} A_{11} & A_{12} \\ A_{21} & A_{22} \end{pmatrix},\] where $A_{11}$ is a $(n-r)\times(n-r)$ matrix, $A_{12}$ is a $ (n-r)\times r$ matrix, $A_{21}$ is a $r\times (n-r)$ matrix, and $A_{22}$ is a $r\times r$ matrix. This yields
    \begin{equation}
      \begin{aligned}\label{eq:reduced_limit}
        \dot{\hat x}  &=\begin{pmatrix} A_{11} & A_{12}  \end{pmatrix}\begin{pmatrix} \hat x \\ \tilde x \end{pmatrix} =\begin{pmatrix} A_{11} & A_{12}  \end{pmatrix}\begin{pmatrix} \hat x \\ T-W' \hat x \end{pmatrix} \\
        &= (A_{11}-A_{12}W')\hat x +A_{12}T \\
        &= C\hat x+DT, 
      \end{aligned}
    \end{equation}
    with $C=A_{11}-A_{12}W'$ and $D=A_{12}$. We call this as the {\em reduced ODE system}. Note that this reduced system has only $n-r$ variables and that the conservation laws are built directly into it. This implies that the differential equation changes if $T$ is changed. The role of this construction is so that we can work only with the non-zero eigenvalues of the $A$.

    As we also have $WA(t)=0$ for all $t\ge 0$, the ODE $\dot x=A(t)x$ can also be similarly reduced to
    \begin{equation}\label{eq:reduced}
      \dot{\hat x}=C(t)\hat x+D(t)T,
    \end{equation}    
    with $C(t)=A_{11}(t)-A_{12}(t)W'$ and $D(t)=A_{12}(t)$, where analogous to $A_{11}$, we define $A_{11}(t)$ to be the top-left $(n-r)\times(n-r)$ sub-matrix of $A(t)$ and  analogous to $A_{12}$, we define $A_{12}(t)$ to be the top-right $ (n-r)\times r$ sub-matrix of $A(t)$.

    Now if $a$ is the equilibrium of the ODE system $\dot{x}=Ax$ starting at $x(0)$, then $\hat a=(a_1,\ldots,a_{n-r})$ is an equilibrium of the reduced ODE sytem $\dot{\hat x}= C\hat x+DT$ starting at $\hat x(0)=(x_1(0),\ldots,x_{n-r}(0))$. Suppose we are able to prove that solutions of reduced ODE $\dot{\hat x}=C(t)\hat x+D(t)T$ starting at $\hat x(0)$ converges exponentially fast towards $\hat a$ then because of the conservation laws $\tilde x=T-W' \hat x$, we would also have that solutions of $\dot x = A(t)x$ starting at $x(0)$ converges exponentially fast towards $a$. So henceforth, we will work only with the reduced ODE systems. For notational convinience, we will drop the hats off $\hat x$ and $\hat a$ and simply refer to them as $x$ and $a$ respectively.

    By subtracting and adding terms to the reduced ODE system (in equation~\ref{eq:reduced}), we have
    \begin{align*}
      \dot{x} &=C(t)x +D(t)T \\
              &= C x +DT+(C(t)-C)x + (D(t)-D)T.
    \end{align*}
    As $a$ is an equilibrium of the ODE system $\dot{x}=C x+DT$, we have $C a+DT=0$. 
    
    Define $y=x-a$.  Then 
    \begin{align*}
      \dot{y} &= C x +DT+(C(t)-C)x + (D(t)-D)T \\
              &=C y +C a +DT+(C(t)-C)x + (D(t)-D)T \\
              &=C y +(C(t)-C)x(t) + (D(t)-D)T \\
              &= C y +E(t)
    \end{align*}
    where it is used that $Ca+DT=0$, and $E(t)=(C(t)-C)x(t) + (D(t)-D)T$.
    
    The solution to the above ODE system is known to be
    \begin{equation}\label{eq:solution}
      y(t)= e^{C t}y(0)+\int_0^t e^{C(t-s)}E(s)\,ds.
    \end{equation}
    We have, using \eqref{eq:solution} and the triangle inequality,
    \[\norm{y(t)}\le \norm{e^{C t}y(0)}+\int_0^t \norm{e^{C(t-s)}}\norm{E(s)} ds.\]
    Now $A$ as defined (see equation~\ref{eq:matrixA} with $k$ inserted for $k(t)$) would form a Laplacian matrix over a strongly connected graph and so it follows that all the eigenvalues of $A$ are either zero or have negative real part. And using $C=A_{11}-A_{12}W'$ and Lemma \ref{lem:kernel} it follows that all eigenvalues of $C$ have negative real part. Hence it follows that 
    \[\norm{e^{C t}}\le  K_2 e^{-\gamma_2 t},\]
    where $0<\gamma_2<-\Re(\lambda_1)$ and $K_2>0$. Here $\lambda_1$ is the eigenvalue of $C$ with the largest real part.
    
    The matrices $C(t)$ and $D(t)$ are linear in $k(t)$. And as $k(t)$ converges exponentially fast towards $k$, it follows that the matrices $C(t)$ and $D(t)$  converge exponentially fast towards $C$ and $D$ respectively. Hence it follows that 
    \begin{align*}
      \norm{E(t)} &= \norm{(C(t)-C)x(t) + (D(t)-D)T} \\
      &\le \norm{(C(t)-C)}\norm{x(t)}+ \norm{(D(t)-D)}\norm{T} \\
      &\le K_3 e^{-\gamma_3 t}+ K_4 e^{-\gamma_4 t}\\
      &\le K_5 e^{-\gamma_5 t}
    \end{align*}
    where
    \begin{itemize}
    \item $\norm{(C_0(t)-C)}\norm{x(t)}\le K_3 e^{-\gamma_3 t}$ for some $K_3,\gamma_3>0$  as $C(t)$ converges exponentially fast towards $C$ and $x(t)$ is bounded (as in the original ODE $\sum_{i=1}^n x_i$ is conserved), and
    \item $\norm{(D_0(t)-D)}\norm{T}\le K_4 e^{-\gamma_4 t}$ for some $K_4,\gamma_4>0$ as $D(t)$ converges exponentially fast towards $D$, and
    \item $K_5=\frac{1}{2}\max(K_3, K_4)$ and $\gamma_5=\min(\gamma_3,\gamma_4)$.
    \end{itemize}

    Collecting all terms we have for all $t\ge 0$,
    \begin{align*}
      \norm{y(t)} &\le \norm{y(0)}K_2e^{-\gamma_2 t}+\int_0^t  K_2 e^{-\gamma_2 (t-s)}\times K_5 e^{-\gamma_5 s}ds \\
      &\le K_0e^{-\gamma_0 t}+K_0\int_0^t  e^{-\gamma_0 t}ds\\
      &= K_0e^{-\gamma_0 t}(1+t)\\
      &\le K e^{-\gamma t} 
    \end{align*} 
    by choosing $K_0=\max\left( K_2K_5, \norm{y(0)}K_2\right )$ and $\gamma_0=\min(\gamma_1,\gamma_2)$. In the last line $\gamma$  is chosen such that $0<\gamma<\gamma_0$ and $K$ is sufficiently large. Since $y(t)=x(t)-a$ we have, \[\norm{x(t)-a}\le  K e^{-\gamma t},\] 
    as required.

  \qed\end{proof}

  We will now prove Theorem~\ref{th:EMexpconv1}. For the sake of convenience, we first recall the statement.\\\
  \par\noindent\textbf{Theorem~\ref{th:EMexpconv1}.}
  {\em For the {Baum Welch Reaction System} $(BW(\mathcal{M},h^*,v^*,L),k)$ with permissible rates $k$, if the concentrations of $\theta$ and $\psi$ species are held fixed at a {\em positive} point then the Forward, Backward and Expection step reaction systems on $\alpha,\beta,\gamma$ and $\psi$ species converge to equilibrium exponentially fast.}
  \begin{proof}\label{pf:EMexpconv1}
    It follows by repeated use of Lemma~\ref{lem:expconv}. For $l=1$ the forward reaction network can be interpretated as the following molecular reactions:
    \begin{equation*}
      \begin{gathered}
        \alpha_{1h}\ce{->[\pi_{h^*}\psi_{h^*w}E_{1w}]} \alpha_{1h^*}\\
        \alpha_{1h^*}\ce{->[\pi_{h}\psi_{hw}E_{1w}]} \alpha_{1h}
      \end{gathered}
  \end{equation*}
  for all $h\in H\setminus\{h^*\}$ and $w\in V$, as they are dynamically equivalent. Here the effective rate constants ($\pi_{h^*}\psi_{h^*w}E_{1w}$ or $\pi_{h}\psi_{hw}E_{1w}$) are independent of time and so the conditions of Lemma~\ref{lem:expconv} are fulfilled. Thus this portion of the reaction network converges exponentially fast.

  The rest of the forward reaction network can be similarly interpretated as the following molecular reactions:
  \begin{equation*}
    \begin{gathered}
      \alpha_{l h}\ce{->[\alpha_{l-1,g}\theta_{gh^*}\psi_{h^*w}E_{lw}]} \alpha_{l h^*}\\
      \alpha_{l h^*}\ce{->[\alpha_{l-1,g}\theta_{gh}\psi_{hw}E_{lw}]} \alpha_{l h}
    \end{gathered}
  \end{equation*}
  for all $g\in H$, $h\in H\setminus\{h^*\}, l=2,\ldots, L$ and $w\in V$. For layers $l=2,\ldots,L$ of the forward reaction network we observe that the effective rate constants ($\alpha_{l-1,g}\theta_{gh^*}\psi_{h^*w}E_{lw}$ or $\alpha_{l-1,g}\theta_{gh}\psi_{hw}E_{lw}$) for layer $l$ depend on time only through  $\alpha_{l-1,g}$. If we suppose that the concentration of $\alpha_{l-1,g}$ converges exponentially fast, then we can use Lemma~\ref{lem:expconv} to conclude that the concentration of $\alpha_{lh}$ also converges exponentially fast. Thus using Lemma~\ref{lem:expconv} inductively layer by layer we conclude that forward reaction network converges exponentially fast. The backward reaction network converges exponentially fast, similarly.
  
  For the expectation reaction network it likewise follows by induction. But here, notice if we interpreted the expectation network similarly into molecular reactions, the effective rate constants would depend on time through the products such as $\alpha_{lh}\beta_{lh}$ or $\alpha_{lh}\beta_{l+1,h}$. So to apply Lemma~\ref{lem:expconv} we need the following: If $\alpha_l(t)$ and $\beta_l(t)$ converge exponentially fast towards $a_l$ and $b_l$ then the product $\alpha_l(t)\beta_l(t)$ converges exponentially fast towards $a_l b_l$ as
  \begin{align*}
    \norm{\alpha_l(t)\beta_l(t)-a_l b_l}&=\norm{(\alpha_l(t)-a_l)(\beta_l(t)-b_l)+a_l( \beta_l(t)-b_l)+b_l( \alpha_l(t)-a_l)} \\
                                        &\le \norm{\alpha_l(t)-a_l}\norm{\beta_l(t)-b_l}+K \norm{\beta_l(t)-b_l}+K \norm{\alpha_l(t)-a_l},
  \end{align*}
  where $K$ is some suitably large constant. We can further observe that $\alpha_l(t)\beta_l(t)$ converges exponentially fast towards $a_l b_l$ at rate $\gamma=\min(\gamma_a,\gamma_b)$, where $\gamma_a$ and $\gamma_b$, respectively, are the exponential convergence rates of $\alpha_l(t)$ and $\beta_l(t)$. A similar argument goes for the products of the form $\alpha_l(t)\beta_{l+1}(t)$. And thus the expectation reaction network, also converges exponentially fast.
\qed\end{proof}

We will now prove Theorem~\ref{th:EMexpconv2}. For the sake of convenience, we first recall the statement.\\\
\par\noindent\textbf{Theorem~\ref{th:EMexpconv2}.}
{\em For the {Baum Welch Reaction System} $(BW(\mathcal{M},h^*,v^*,L),k)$ with permissible rates $k$, if the concentrations of $\alpha,\beta,\gamma$ and $\xi$ species are held fixed at a {\em positive} point then the Maximization step reaction system on $\theta$ and $\psi$ converges to equilibrium exponentially fast.}
\begin{proof}\label{pf:EMexpconv2}
  Exponential convergence of the maximisation network follows by a similar layer by layer inductive use of Lemma~\ref{lem:expconv}. 
\qed\end{proof}

\end{appendices}
\end{document}